\documentclass{article}

\usepackage[english]{babel}
\usepackage[utf8]{inputenc}
\usepackage[T1]{fontenc}
\usepackage{lmodern}
\usepackage{graphicx}
\usepackage{framed}
\usepackage[a4paper, margin=1in]{geometry}
\usepackage[framemethod=tikz]{mdframed}
\usepackage{todonotes}
\usepackage{color}

\definecolor{darkgreen}{rgb}{0,0.5,0}
\definecolor{darkgray}{rgb}{0.2,0.2,0.2}

\usepackage{hyperref}
\hypersetup{
    unicode=false,          
    colorlinks=true,        
    linkcolor=blue,          
    citecolor=darkgreen,        
    filecolor=magenta,      
    urlcolor=cyan           
}

\usepackage{amsthm}
\usepackage{amsmath}
\usepackage{amsfonts}
\usepackage{mathrsfs}
\usepackage{mathtools}
\usepackage{verbatim}
\usepackage{footnote}
\usepackage{xspace}
\usepackage{cleveref}

\usepackage{algorithm}
\usepackage[noend]{algpseudocode}
  
\usepackage{lineno}
\usepackage{caption}
\usepackage{framed}
\usepackage{enumerate}
\usepackage{enumitem}

\usepackage{tikzsymbols}
\usepackage{thmtools,thm-restate}
\usepackage{nicefrac}

\newtheorem{theorem}{Theorem}
\newtheorem{lemma}[theorem]{Lemma}

\newtheorem{definition}[theorem]{Definition}

\newtheorem{remark}{Remark}
\newtheorem*{remark*}{Remark}

\crefname{theorem}{Theorem}{Theorems}
\Crefname{lemma}{Lemma}{Lemmas}
\Crefname{invariant}{Invariant}{Invariants}
\Crefname{claim}{Claim}{Claims}
\Crefname{observation}{Observation}{Observations}
\Crefname{algorithm}{Algorithm}{Algorithms}
\Crefname{figure}{Figure}{Figures}
\Crefname{challenge}{Challenge}{Challenges}
\Crefname{problem}{Problem}{Problems}

\def\alloc{\mathsf{alloc}}
\def\vx{\mathsf{x}}

\def\mcap{\mathsf{C}}
\def\neigh{\mathsf{N}}
\def\halloc{\overline{\mathsf{alloc}}}

\def\mweight{\mathsf{MatchWeight}}
\def\gap{\frac{2\log(2|R| / \epsilon)}{\epsilon} + 1}
\def\hbeta{\overline{\beta}}
\def\hvx{\overline{\vx}}

\newcommand{\logn}{\frac{\log(|R| / \epsilon)}{\epsilon^2}}
\newcommand{\lognExact}{\frac{2\log(2|R| / \epsilon)}{\epsilon^2} + \frac{1}{\epsilon}}
\newcommand{\loglambda}{\log_{1+\epsilon}(4\lambda / \epsilon) + 1}

\newcommand{\eps}{\epsilon}

\newcommand{\cL}{\mathcal{L}}
\newcommand{\tO}{\tilde{O}}

\newcommand{\rb}[1]{\left( #1 \right)}
\DeclareMathOperator{\poly}{poly}
\newcommand{\bbN}{\mathbb{N}}

\def\poly{\text{poly}}

\begin{document}

\title{Faster MPC Algorithms for Approximate Allocation in Uniformly Sparse Graphs} 

\author{ Jakub Łącki \thanks{e-mail: \texttt{jlacki@google.com}} \\ Google Research
\and Slobodan Mitrović\thanks{Supported by the Google Research Scholar and NSF Faculty Early Career Development Program No.~2340048. e-mail: \texttt{smitrovic@ucdavis.edu}} \\ UC Davis 
\and Srikkanth Ramachandran\thanks{Supported by NSF Faculty Early Career Development Program No.~2340048. e-mail: \texttt{sramach@ucdavis.edu}} \\ UC Davis
\and Wen-Horng Sheu\thanks{Supported by NSF Faculty Early Career Development Program No.~2340048. e-mail: \texttt{wsheu@ucdavis.edu}} \\ UC Davis}
\date{ }



\maketitle

\begin{abstract}
We study the allocation problem in the Massively Parallel Computation (MPC) model. This problem is a special case of $b$-matching, in which the input is a bipartite graph with capacities greater than $1$ in only one part of the bipartition.
We give a $(1+\epsilon)$ approximate algorithm for the problem, which runs in $\tilde{O}(\sqrt{\log \lambda})$ MPC rounds, using sublinear space per machine and $\tilde{O}(\lambda n)$ total space, where $\lambda$ is the arboricity of the input graph.
Our result is obtained by providing a new analysis of a LOCAL algorithm by Agrawal, Zadimoghaddam, and Mirrokni~[ICML 2018], which improves its round complexity from $O(\log n)$ to $O(\log \lambda)$.
Prior to our work, no $o(\log n)$ round algorithm for constant-approximate allocation was known in either LOCAL or sublinear space MPC models for graphs with low arboricity.
\end{abstract}

\newpage

\section{Introduction}
The maximum matching problem is one of the most central problems studied in the Massively Parallel Computation (MPC) model.
Indeed, several highly influential lines of research in the MPC model started with the study of the maximum matching problem, including the introduction of the filtering technique~\cite{lattanzi2011filtering}, the round compression technique~\cite{czumaj2018round,GU19}, or the component stability framework~\cite{GKU19, CDP21}.
The extensive research on the problem has led to a number of improved MPC algorithms for maximum matching in several settings~\cite{lattanzi2011filtering,czumaj2018round,GGK+18,assadi2019coresets,BBD+19,behnezhad2019exponentially,GU19,GGJ20,ghaffari2022massively,dhulipala2024parallel}.

The MPC model is a theoretical abstraction of large-scale frameworks, such as MapReduce, Flume, Hadoop, and Spark, which was introduced in a series of papers~~\cite{dean2008mapreduce,karloff2010model,goodrich2011sorting}. 
An instance of this model has $N$ machines, each equipped with $S$ words of space.
The computation on these machines proceeds in synchronous rounds. 
Originally, the input data is partitioned across the machines arbitrarily.
During a round, a machine performs computation on the data it has in its memory.
At the end of a round, the machines simultaneously exchange messages.
Each machine can send messages to any other machine as long as the total number of words sent and received by a machine in one round does not exceed $S$.
For solving a problem on $n$-vertex graphs, the known MPC algorithms focus on one of the following regimes: the \emph{superlinear} regime, in which $S = n^{1 + \gamma}$; the \emph{near-linear} regime, where $S = n \cdot \poly \log n$; and the \emph{sublinear} regime, where $S = n^{\gamma}$. Here $\gamma$ is an arbitrary constant satisfying $0 < \gamma < 1$.
The ultimate goal in the MPC model is to come up with algorithms that use a low number of rounds in the most challenging, sublinear regime.

After several successful years, the progress in developing new maximum matching algorithms in the MPC model has recently slowed down.
In particular, the best known algorithm for $(1+\epsilon)$-approximate maximum matching in the sublinear regime, which uses $\tilde{O}(\sqrt{\log n})$ rounds, has been proposed over 5 years ago~\cite{GU19}, and improving this bound has been a major open problem since.

Several lines of work studied different variants of the problem in order to break the notorious barrier of $\tilde{O}(\sqrt{\log n})$ rounds.
Behnezhad et al. \cite{BBD+19} showed that for graphs of arboricity $\lambda$, approximate maximum matching can be found in $\tilde{O}(\sqrt{\log \lambda} + \log^2 \log n)$ rounds in the sublinear regime, which is an improvement for graphs of low arboricity.
Moreover, in the near-linear regime Ghaffari, Grunau, and Mitrović showed how to solve approximate matching (and even a more general problem of b-matching)  in $O(\log \log n)$ rounds~\cite{ghaffari2022massively}.

We contribute to this line of research by showing that within the best known bounds for computing approximate matching in the sublinear MPC regime, one can solve a more general \emph{allocation} problem.
Given a bipartite graph $G = (L \cup R, E)$ with integer capacities $\mcap : R \to \bbN_{\ge 1}$, the allocation problem aims to find a set of edges $M \subseteq E$ of maximum cardinality, such that a vertex in $L$ is incident to at most $1$ edge and a vertex $v \in R$ is incident to at most $\mcap_v$ edges in $M$.
Particularly, the allocation problem is a generalization of the maximum matching problem in bipartite graphs.
Various versions of this problem have been extensively studied in computer science, especially in the context of online ads and server-client resource allocation~\cite{mehta2007adwords,feldman2009online,vee2010optimal,wang2015socially,AZM18,matching-paper1,matching-paper2,matching-paper3,matching-paper4,matching-paper5,matching-paper6}.
Moreover, the allocation problem is used as a subroutine in other problems~\cite{AZM18}.
In particular, it was used to obtain the state-of-the-art algorithm for load balancing~\cite{ahmadian2021distributed}.

In this paper, we show that the allocation problem can be solved in $\tilde{O}(\sqrt{\log \lambda})$ rounds in the sublinear MPC regime.

\subsection{Our contributions}
Our first contribution is a new analysis of the algorithm of Agrawal, Zadimoghaddam, and
Mirrokni \cite{AZM18}, which works in the LOCAL model.
The LOCAL model is a distributed model of computation for graph problems, in which the communication graph is the same as the input graph.
The computation happens at the vertices of the graph.
The algorithms work in synchronous rounds, where a vertex can send messages of arbitrary size to its neighbors in each round.

Given an instance of the allocation problem, the LOCAL algorithm of Agrawal, Zadimoghaddam, and
Mirrokni computes a \emph{fractional allocation}, which is a function that assigns each edge a real number in $[0, 1]$ while respecting the capacity constraints (see \cref{def:fractional-allocation} for a formal definition).
The analysis of \cite{AZM18} shows that the algorithm computes a $1 + \epsilon$ approximate fractional allocation within $O_{\epsilon}(\log n)$ rounds.
We show that the same algorithm, in fact, converges to a $2+\epsilon$ approximate solution within $O_{\epsilon}(\log \lambda)$ rounds, where $\lambda$ is the arboricity of the graph. 
Our algorithm can be transformed into a constant approximation algorithm for the (integral) allocation problem using standard techniques.
In addition, its approximation factor can be boosted to $1+\eps$ using the framework in \cite{ghaffari2022massively}.
The round complexity remains $O_\eps(\log \lambda)$ but has a higher dependency on $\eps$.

\begin{theorem}\label{thm:LOCAL-main}
    There is a randomized LOCAL algorithm that computes a $1+\epsilon$ approximate (integral) allocation in $O_\eps(\log \lambda)$ rounds. 
\end{theorem}
We focus on describing the $2+\epsilon$ approximation for fractional allocation, i.e., we focus on proving the following claim. The details on obtaining a $(1+\epsilon)$-approximation algorithm are deferred to \cref{sec:rounding,sec:framework}.

\begin{theorem}\label{thm:LOCAL-allocation}
    There is a deterministic LOCAL algorithm that computes a $2+\epsilon$ approximate fractional allocation in $O_\eps(\log \lambda)$ rounds. 
\end{theorem}


\begin{remark}
The intuition behind \cref{thm:LOCAL-allocation} is as follows. Agrawal et al. \cite{AZM18}'s algorithm begins by saturating the densest part of the graph and then gradually spreads the allocation to the sparser regions using a multiplicative weight update approach.
The speed of this process is inversely proportional to the density of the densest part, and thus the algorithm converges faster on low-arboricity graphs.
\end{remark}

We remark that \cref{thm:LOCAL-allocation} cannot be obtained by the existing reduction from the allocation problem to the maximum matching problem.
One such reduction is to split each vertex $v \in R$ into $\mcap_v$ copies such that each copy is adjacent to all neighbors of $v$.
It is not hard to see that solving the allocation problem on $G$ reduces to finding a maximum matching in the new graph.
However, this reduction can significantly increase the arboricity of the graph.
For instance, if $G$ is a star, in which the center has capacity $n - 1$ and all leaves have capacity $1$, the above reduction creates a complete bipartite graph with $n - 1$ vertices in each part.
Thus, the arboricity increases from $1$ to $\Theta(n)$.

Leveraging \cref{thm:LOCAL-allocation}, we give a new algorithm for the allocation problem in the MPC model.
Our MPC algorithm works in the sublinear regime, which is the most restrictive, and thus the most challenging one.
In all our round complexities we hide dependency on $\gamma$ within the $O(\cdot)$ notation.


To obtain our MPC algorithm, we show how to sparsify the computation graph for the allocation problem, which enables us to simulate $\Omega(\sqrt{\log \lambda})$ LOCAL rounds of the algorithm of \cref{thm:LOCAL-allocation} in $O(\log \log \lambda)$ MPC rounds.
As a result, we obtain an $O_{\epsilon}(\sqrt{\log \lambda} \cdot \log \log \lambda)$-round MPC algorithm which computes a $(2+\epsilon)$-approximate fractional allocation.
We then use standard techniques to transform this into a $\Theta(1)$-approximate integral solution.
Finally, by applying the recent framework of \cite{ghaffari2022massively} for $b$-matching, we  improve our approximation factor to $1+\epsilon$.

Interestingly, the analysis of \cref{thm:LOCAL-allocation} also enables us to run the algorithm without knowing the value of $\lambda$.
We remark that in the MPC model to the best of our knowledge there is no efficient algorithm for checking whether a given matching is a $(1+\epsilon)$-approximate matching. 
As a result, the standard approach of guessing $\lambda$, i.e., running the algorithm for different choices of $\lambda = 2^i$, does not seem to apply in a straightforward way without increasing the round complexity.
This gives our main result of the paper, which is formally stated below.

\begin{theorem}\label{thm:alloc-main}
    For any constant $\alpha \in (0, 1)$, there is an $O_{\epsilon}(\sqrt{\log \lambda} \cdot \log \log \lambda)$-round MPC algorithm using $n^{\alpha}$ local memory and $\tilde{O}(\lambda n)$ total memory that with high probability computes a $1+\epsilon$ approximate solution for allocation. Moreover, the algorithm does not require knowledge of $\lambda$.
\end{theorem}


The best previously known MPC algorithm for the allocation problem follows by adapting a LOCAL algorithm by \cite{AZM18} to MPC. This adaptation yields an $O(\log n)$ MPC round approach for the $(1+\eps)$-approximate allocation problem.

\subsection{Background and State of the Art}

\subsubsection{The allocation and b-matching problem in MPC}
Agrawal, Zadimoghaddam and Mirrokni showed that the fractional $(1+\eps)$-approximate allocation problem can be solved in $O\rb{\frac{\log n}{\eps^2}}$ rounds of distributed computation~\cite{AZM18}. Since their approach sends only $\poly \log n$-size messages over an edge in each round, the algorithm readily translates to the sublinear MPC regime with the same round complexity.
This algorithm was used as a subroutine in an approximate load balancing framework~\cite{ahmadian2021distributed}.

Several works have studied the complexity of computing $b$-matching in MPC, which is a generalization of the allocation problem. 
The most recent work shows how to find a $(1+\eps)$-approximate maximum $b$-matching in $O_\eps(\log \log n)$ MPC rounds~\cite{ghaffari2022massively} in the near-linear memory regime. The approach of their work yields $O(\log n)$ round complexity for the sublinear memory regime.
In the blackboard distributed computation, \cite{liu2023scalable} study $(1+\eps)$-approximate $b$-matching in bipartite graphs aiming to reduce the dependence on $1/\eps$. 
They obtain $O_\eps(\log n)$ round complexity.
To the best of our knowledge, no algorithm is known for $\Theta(1)$-approximate maximum $b$-matching that runs in $o(\log n)$ or even $o(\log \lambda)$ MPC rounds in the sublinear memory regime, and it remains an open question whether such an algorithm exists.
Our work on the allocation problem can be seen as the first step towards answering that question in the affirmative.

\subsubsection{Matching in MPC}
\cite{lattanzi2011filtering} showed that a maximal matching can be found in $O(1)$ in the super-linear regime. In a more challenging, near-linear memory regime, a line of work~\cite{czumaj2018round,GGK+18,assadi2019coresets,behnezhad2019exponentially, ghaffari2022massively} resulted in a $O(\log \log \bar{\Delta})$ MPC round complexity for computing maximal matching, where $\bar{\Delta}$ is the average degree.
In the sublinear memory regime, it is known how to compute a maximal matching in $\tO(\sqrt{\log \Delta})$ MPC rounds~\cite{GU19}.

Several authors focused on obtaining faster algorithms for graphs with low arboricity. In particular, \cite{BBD+19} show how to compute a maximal matching in $O(\log^2 \log n + \sqrt{\log \lambda} \log \log \lambda)$ many rounds. That result was improved by \cite{GGJ20} to $O(\log \log n + \sqrt{\log \lambda} \log \log \lambda)$ round complexity.

\subsection{Organization}
\cref{sec:prelim} provides necessary notations and definitions.
\cref{sec:overview} gives an overview of our techniques.
In \cref{sec:local}, we describe our LOCAL algorithm for computing a $(2+\eps)$-approximate fractional allocation.
Then, we show in \cref{sec:mpc} how to simulate the algorithm in the MPC model.
\cref{sec:rounding} outlines the rounding procedure to transform a fractional allocation into an integral one.
\cref{sec:mpc:correctness} provides correctness proof for our MPC algorithm.
\cref{sec:framework} outlines how to boost the approximation factors of our algorithms to $1 + \eps$, using \cite{ghaffari2022massively}'s framework.

\section{Preliminaries and notation} \label{sec:prelim}
\subsection{Basic notations and definitions}
In this paper, we deal with a bipartite graph $G = (V, E)$ with the bipartition $V = L \cup R$. 
Let $n$ and $m$ be, respectively, the number of vertices and edges of $G$.
We use $\neigh_{G, v}$ to refer to the neighborhood of $v$ in the graph $G$. We shall drop $G$ when it is clear in context.

\begin{definition}[Arboricity]\label{def:arb}
    The arboricity of a graph $G$ is the minimum number of forests into which its edges can be partitioned.
    We denote the arboricity of $G$ by $\lambda(G)$ or simply $\lambda$ when $G$ is clear in context.
\end{definition}

Every subgraph of a graph $G$ has average degree $O(\lambda(G))$. Hence, small arboricity implies that the graph is ``everywhere sparse''. In the literature, graphs with small arboricity $\lambda$ are often referred to as \emph{uniformly sparse}.

\def\seta{\mathbb{A}}
\def\seti{\mathbb{I}}

\begin{definition}[Allocation]\label{def:alloc-prob}
    Let $G$ be a bipartite graph with bipartition $L \cup R$. Let $\mcap : R \rightarrow \mathbb{N}$ be a capacity function. An \emph{allocation} is a subset $M \subseteq E(G)$ of edges such that (i) every $u \in L$ has at most one incident edge in $M$ and (ii) every $v \in R$ has at most $\mcap_v$ incident edges in $M$. The \emph{allocation problem} asks for an allocation of the maximum size.
\end{definition}

\noindent We note that the allocation problem is a special case of the $b$-matching problem for bipartite graphs wherein all vertices on one side of $G$ have $b$-value as $1$.

\begin{definition}[Fractional allocation] \label{def:fractional-allocation}
    Let $(G, L, R, \mcap)$ be an instance of the allocation problem. A \emph{fractional allocation} is a function that assigns, to each edge $(u, v)$ in $E(G)$, a fraction $\vx_{u, v} \in [0, 1]$ such that (i) every $u \in L$ satisfies $\sum_{v \in \neigh_{u}} \vx_{u, v} \leq 1$ and (ii) every $v \in R$ satisfies $\sum_{u \in \neigh_{v}} \vx_{u, v} \leq \mcap_v$. The \emph{weight} of the fraction allocation is defined as $\sum_{(u, v) \in E(G)} \vx_{u, v}$. A \emph{maximum fractional allocation} is a fractional allocation of the maximum weight.
\end{definition}

\subsection{LOCAL model} The $\text{LOCAL}$ model is one of the most standard models in distributed computing for studying graph problems.
In this model, the communication network is defined by the input graph $G = (V, E)$, where each vertex is associated with a unique identifier of $O(\log n)$ bits.
There is a processor with unbounded memory in each graph vertex.
The processor in vertex $v$ initially only knows the unique identifier of $v$.
The processors communicate along the edges of the graph in synchronous rounds.
In each round, each processor:
\begin{enumerate}
    \item Receives messages sent to it in the previous round.
    \item Performs arbitrary computation.
    \item Sends messages of arbitrary size to the processors in neighboring vertices. These messages are delivered at the beginning of the next round.
\end{enumerate}

The goal in this model is to solve a given problem in the smallest number of rounds.
Note that if the graph diameter is $D$, each processor can learn the entire graph in $D$ rounds, and thus, solving any graph problem in $D$ rounds is trivial.

\subsection{MPC model} The $\text{MPC}$ model is the standard model for modern parallel computing. In this model, the communication network is a clique.
A graph $G$ of $n$ vertices is stored across the network.
Each machine in the network contains limited memory of $S$ bits.
Communication occurs in synchronous rounds. In each round, a machine can send and receive a total of $S$ words of information. Machines may perform \textit{arbitrary} computation within a single round. We refer to the global memory as the total space across all machines in the network. Problems are naturally harder to solve with limited space and global memory. In this paper, we deal with the \textit{sublinear} space regime, i.e., $S = n^\alpha$ for some constant $\alpha \in (0, 1)$. In all our round complexities we hide dependency on $\alpha$ within the $O(\cdot)$ notation. Our algorithms require $\tilde{O}(\lambda n)$ global memory.

\section{Overview of techniques} \label{sec:overview}

\subsection{Review of the algorithm by \texorpdfstring{\cite{AZM18}}{Ahmadian et al.}}

\begin{algorithm}
\begin{algorithmic}[1]
\Statex \textbf{Input:} $G = (L \cup R, E)$, capacity constraints $\{\mcap_v\}_{v \in R}$, number of rounds $\tau \geq 1$, approximation parameter $\epsilon > 0$.
\Statex \textbf{Output:} a fractional allocation $\{\vx'_{u,v}\}_{(u,v) \in E}$ 
\Statex \textbf{Initialization:} Set $\beta_v \gets 1$ for every $v \in R$

\For{rounds $r = 1, 2, \dots \tau$}
    \State for each vertex $u \in L$: \label{alg-local:agg1} 
    \Statex \hspace{1cm} Set $\vx_{u, v} = \frac{\beta_v}{\sum_{v' \in \neigh_u} \beta_{v'}}$ for each $v \in \neigh_u$ 
    \State for each vertex $v \in R$: \label{alg-local:agg2} 
    \Statex \hspace{1cm} Set $\alloc_v \gets \sum\limits_{u \in \neigh_v} \vx_{u, v}$
    \State for each vertex $v \in R$, update  $\beta_v$ as follows, \label{alg-local:thresh-checker}
            
    \Statex  $$ \beta_v \gets \begin{cases} 
        \beta_v (1 + \epsilon) \ \ \ \text{if } \alloc_v \leq \mcap_v / (1 + \epsilon) \\
        \beta_v / (1 + \epsilon) \ \ \text{if } \alloc_v \geq \mcap_v (1 + \epsilon) \\
        \beta_v \ \ \ \ \ \ \ \ \ \ \ \ ~~~\text{otherwise}      
    \end{cases} $$    
\EndFor

\For {each vertex $v \in R$}
    \State for all $u \in \neigh_v$, set $\vx'_{u,v} \gets \begin{cases} 
        \frac{\mcap_v}{\alloc_v}\vx_{u,v} \ \ \ \text{if } \alloc_v > \mcap_v \\
        \vx_{u,v} \ \ \ \ \ \ \ \ \ \ \ \ \text{otherwise}      
    \end{cases}$
\EndFor \label{line:scale}
\end{algorithmic}
\caption{A LOCAL algorithm for finding a fractional allocation}
\label{alg:local}
\end{algorithm}

We first briefly overview the main idea of the \textit{proportional allocation algorithm} by \cite{AZM18}, which is shown in \cref{alg:local}.

The input to \cref{alg:local} is a bipartite graph with a known bipartition $(L, R)$, capacity constraints $\{\mcap_v\}_{v \in R}$, and two parameters $\tau \geq 1$ and $\epsilon > 0$.
The output is a fractional allocation.
For each vertex $v \in R$, the algorithm maintains a \emph{priority value} $\beta_v$, initially set to $1$. 
Each edge $(u, v)$ is assigned a fractional value $x_{u,v}$. The algorithm performs $\tau$ rounds of computation to update the fractional values. (See lines 1-4.) In each round, each vertex $u \in L$ calculates the fractional value $\vx_{u, v}$ assigned to each vertex $v \in \neigh_u$; the value is computed as the ratio between $\beta_v$ and the sum of $\beta$ values over all vertices in $\neigh_v$. We remark that the fractional values computed in this step may not form a fractional allocation because the sum of all fractions allocated to a vertex $v \in R$ may exceed $\mcap_v$. After this step, each vertex $v \in R$ examines $\alloc_v = \sum_{u \in \neigh_v} \vx_{u, v}$ to update its priority. (See line 4.)
The update is based on the following rules: if $v$ is \textit{under-allocated} by a factor of $(1 + \epsilon)$, meaning $\alloc_v \leq \mcap_v / (1 + \epsilon)$, then $\beta_v$ is increased by a factor of $(1 + \epsilon)$. 
Symmetrically, if $v$ is \textit{over-allocated} by a factor of $(1 + \epsilon)$, $\beta_v$ is decreased. 
If neither of the above two cases holds, $\beta_v$ stays unchanged this round. After $\tau$ rounds of computation, lines 5-6 transform $\{ \vx_{u, v} \}_{(u,v) \in E}$ to a fractional allocation by scaling the values assigned to each vertex $v \in R$.

\def\lset{\mathcal{L}}

Agrawal et al.~showed that the algorithm finds a $(1 + O(\epsilon))$-approximate fractional allocation if $\tau \geq O(\logn)$. Their proof is based on two crucial observations. 
Consider the state after the end of the $\tau$ rounds. 
Partition $R$ into $2\tau+1$ \textit{level sets} $\lset_0, \lset_1, \dots, \lset_{2\tau}$, where $\lset_i = \{v \in R \mid \beta_v = (1 + \epsilon)^{(i - \tau)}\}$. 

The first observation is that each vertex $v \in R$ with $v \notin \lset_{2\tau}$ has filled up a $(1 - O(\epsilon))$ fraction of its capacity; in addition, each vertex $v \in R$ with $v \notin \lset_0$ is over-allocated by at most a factor of $1 + O(\epsilon)$, which implies that lines 5-6 only scale $\vx_{u,v}$ by a $\frac{1}{1 + O(\epsilon)}$ factor. This observation provides a lower bound on the output fractional allocation. 

The second observation is that given any vertex subset $S = \bigcup_{j \leq \ell} \lset_j$ for some $\ell \leq 2\tau$, the cardinality of the optimal fractional allocation is at most the total capacity of $S$ plus the size of the neighborhood of $R \setminus S$. Agrawal et al. proved the approximation factor by showing that the lower bound in the first observation almost matches one of the upper bounds in the second observation.

\subsection{Our approach}

To obtain our result, a key finding is that \cref{alg:local} after $\loglambda$ rounds outputs a $(2 + O(\epsilon))$-approximate allocation.
By observations discussed above, the weight of the optimal solution is at most the total capacity of $\lset_0, \lset_1, \dots, \lset_{2\tau-1}$ plus the number of neighbors of $\lset_{2\tau}$. 
Furthermore, all vertices in $\lset_0, \lset_1, \dots, \lset_{2\tau - 1}$ almost fill up their capacity. 
Hence, to obtain a $(2 + O(\epsilon))$-approximation, it suffices to choose a large enough $\tau$ so that $|\neigh(\lset_{2\tau})|$ is upper-bounded by the weight of the output allocation. 
We observe that this holds if almost all fractional values from $\neigh(\lset_{2\tau})$ are allocated to vertices not in $\lset_0$, as these vertices have $1+O(\epsilon)$ over-allocation.

Suppose that the algorithm fails to find a $(2 + O(\eps))$-approximation after some $\tau \geq 1$ rounds, which implies that a large fraction of $\neigh(\lset_{2\tau})$ is allocated to $\lset_0$.
Intuitively, this can only happen if the subgraph induced by $\neigh(\lset_{2\tau}) \cup \lset_0$ is very dense.
The above intuition can be formalized by noticing that each edge from $\neigh(\lset_{2\tau})$ to $\lset_0$ can only allocate a value of at most $1 / (1 + \eps)^{2\tau}$, due to the difference in priority values between $\lset_0$ and $\lset_{2\tau}$.
Hence, to allocate a constant fraction of $\neigh(\lset_{2\tau})$ to $\lset_0$, there must be $\Omega(|\neigh(\lset_{2\tau})| \cdot (1 + \eps)^{2\tau})$ edges between them.
On the other hand, by the property of bounded arboricity, the density of any subgraph is bounded by a function of $\lambda$.
We show that for $\tau \approx \log_{1+\epsilon} (\lambda / \epsilon)$, the above arguments lead to a contradiction, as the dense subgraph induced by $\neigh(\lset_{2\tau}) \cup \lset_0$ cannot exist in a low-arboricity graph.

Our finding provides a $O(\log\lambda)$-round LOCAL algorithm to compute a constant approximate maximum fractional allocation on a bipartite graph with a known bipartition $(L, R)$ and known upper bound on the arboricity. The fractional solution can be converted into an integral one using standard techniques. (See \cref{sec:rounding} for details.) In addition, the approximation factor can be boosted to $1+\eps$ using the framework of \cite{ghaffari2022massively}. (See \cref{sec:framework}.)




\subsubsection{Improved round complexity in MPC}
    A straightforward simulation of \cref{alg:local} in sub-linear MPC would require $O(\log \lambda)$ rounds. In this section we describe how to improve this complexity by adapting the graph sparsification technique introduced by \cite{GU19} in a nontrivial way. 

    The general idea is to split the execution of $R$ rounds of a LOCAL algorithm into $R / B$ phases of $B$ rounds each and simulate each $B$ rounds within a phase in $o(B)$ MPC rounds.
    Observe that in order to simulate $B$ LOCAL rounds in MPC, for each vertex $v$, it suffices to collect $v$'s $B$-hop neighborhood on a single MPC machine.
    Subsequently, the entire phase consisting of $B$ LOCAL rounds can be simulated without additional communication. 
    Collecting subgraphs can be done efficiently in $O(\log B)$ rounds using the graph exponentiation technique~\cite{lenzen2010brief,GU19}, i.e., by doubling the radius every round. 
    This approach yields a $O(R/B \cdot \log B)$ round MPC algorithm, \emph{assuming} that the $B$-hop neighborhood of each vertex fits into the memory of a single machine. 
    However, since vertices can have large degrees, the size of a $B$-hop neighborhood may be larger than the available memory per machine (even for $B=1$), which makes a direct implementation of this approach infeasible.
    
    To circumvent this, we redesign our LOCAL algorithm such that no communication occurs along many edges during $B$ rounds. 
    These edges are then ignored when collecting the $B$-hop neighborhoods, enabling our algorithm to leverage the aforementioned graph exponentiation technique.
    Our redesign of \cref{alg:local} stems from two observations:
    \begin{enumerate}
        \item A value $\beta_v$ changes by a factor of at most $1+\epsilon$ in a single round.
        \item Estimating $\beta_v$ is enough for \cref{alg:local} to be correct. In particular, a simple modification of the update rule by replacing $1+\epsilon$ to $1+k\epsilon$ for some bounded constant $k$ would still yield a $O(1)$-approximation. We describe the modified algorithm and analyze it in \cref{sec:mpc:correctness}.
    \end{enumerate}

    Based on Observations~1 and 2, we arrive at the following idea: instead of computing $\beta_v$ by aggregating over the entire neighborhood of $v$ in each round (as in \cref{alg:local}), sample a subset of the neighborhood and estimate the desired aggregations by extrapolating from the samples. 
    If we have enough samples so that $\beta$ is approximated to within a constant, then by Observation~2, our algorithm is correct.

    We outline how these two observations are incorporated into our MPC algorithm.
    Suppose we wish to simulate rounds $r+1, \dots r+B$ of \cref{alg:det-thresh}.
    Assume that we have previously  simulated the first $r$ rounds, and so every vertex $v$ knows its $\beta_v$ after $r$ rounds. 
    For each vertex $v$,
let $\lset^s_v(i)$ denote the set of neighbors $w$ of $v$ such that $\beta_w \in [(1+\epsilon)^{i-1}, (1+\epsilon)^i)$ at the end of round $s$.
    
    The update rules in \cref{alg:local} require the sum of $\beta$ values over the neighbors of $v$.
    To approximate these sums, motivated by Observation~1, our algorithm uniformly at random samples a subset of $\lset^r_v(i)$ for every $i$ and every $v$. 
    Since within $B$ rounds of the LOCAL algorithm, the $\beta$ value of a vertex can change by $(1+\eps)^B$ at most, sampling $O((1+\eps)^{O(B)} \cdot \poly \log n)$ 
    vertices from $\lset^r_v(i)$ allows to approximate the total sum of the $\beta$ values of the vertices in $\lset^r_v(i)$ in each of the next $B$ rounds. To simulate each round $s$ (for $s \in [r, r+B]$), we use \textit{different} independent samples from $\lset^r_{v}(i)$. The $\poly \log$ factor ensures that we have enough samples to use, as there are at most $B < \log n$ rounds.
    We emphasize here that our analysis is based on showing concentration for our estimates of the sum of $\beta$ values within $\lset^r_v(i)$, as opposed to within $\lset^q_v(j)$ for $q = r + 1 \ldots r + B$, over the next $B$ rounds.
    
    We make this intuition formal in \cref{lem:sampling}, by proving that $O((1+\epsilon)^{2B} \log n)$ samples from each set $\lset^r_v(i)$ guarantees a $1+O(\epsilon)$ approximation, leading us to \cref{alg:matching}.

    \textbf{Note on memory requirement.} The size of the subgraph around each vertex is $2^{O_{\epsilon}(B^2)}$. 
    Every vertex $v$ needs to participate in the round executions; therefore, we require $n \cdot 2^{O_{\epsilon}(B^2)}$ total memory. The maximum possible value of $B$ is thus $O_{\epsilon}(\alpha \sqrt{\log n})$, in which case the total required memory is $\tilde{O}(n^{1+\alpha} + m)$. 
    Setting $B = O_{\epsilon}(\sqrt{\log \lambda})$, we get the total required memory of $\tilde{O}(\lambda n)$.

\subsubsection{Removing the assumption on known arboricity in MPC}

Our algorithm depends on the knowledge of $\lambda$ in two ways:
    \begin{enumerate}
        \item The condition in Line~1 of \cref{alg:local} assumes the knowledge of $\tau$, which is computed based on $\lambda$.
        \item In the graph exponentiation approach, we need the value of $\sqrt{\log \lambda}$ to know how many rounds to simulate.
    \end{enumerate}

In order to remove these dependencies, we first observe that our algorithm can detect whether sufficiently many rounds have been run without knowing $\lambda$.
Specifically, we show that in $O(\log \lambda)$ rounds, one of the following conditions must hold: either the size of $|\neigh(\lset_{2\tau})|$ becomes smaller than $|\lset_0|$, or a $\frac{1}{1 + O(\epsilon)}$ fraction of $|\neigh(\lset_{2\tau})|$ is allocated to vertices not in $\lset_0$. 
Moreover, we prove that if one of the conditions holds, and we terminate the algorithm, the output is a $(2 + O(\epsilon))$-approximation.
We note that both conditions can be tested for in $O(1)$ MPC rounds.

Thanks to the above observation, we can ``guess'' the value of $\lambda$.
More precisely, we first run the algorithm assuming that $\lambda$ is a small constant and check whether the termination condition holds after $O(\log \lambda)$ rounds.
If yes, we know that we have obtained a $(2 + O(\epsilon))$-approximate solution.
If not, we increase the value of $\lambda$ and repeat.
We choose the value of $\lambda_i$ which we use in the $i$th trial, such that $\sqrt{\log \lambda_i} = 2^i$.
This ensures that the running time is only a constant factor larger than what we would have spent if we knew $\lambda$ upfront.

\section{Approximating maximum fractional allocation in low arboricity graphs} \label{sec:local}

The objective of this section is to show that \cref{alg:local} has an approximation factor of $2 + 10\epsilon$ if $\tau \geq \loglambda$.
Let $\mweight = $ $\sum_{(u,v) \in E} \vx'_{u,v}$ be the weight of the output fractional allocation.
Note that $\mweight = \sum_{v \in R} \min(\mcap_v, \alloc_v)$.

Consider the $\beta$ variables after the end of round $\tau$.
The minimum value the priority variables can take after $\tau$ rounds is $\beta_{min} = \frac{1}{(1+\epsilon)^\tau}$, and any $v \in R$ can take one of the following $2\tau + 1$ possible priority values:

\begin{equation*}
    \beta_v \in \{\beta_{min}, (1+\epsilon)\beta_{min}, \dots, (1+\epsilon)^{2\tau}\beta_{min}\}
\end{equation*}
For each $0 \leq j \leq 2\tau$, let $\lset_j$ be the set of vertices in $R$ with priority value $(1 + \epsilon)^j \beta_{min}$, that is, $\lset_j = \{ v \mid v \in R \mbox{ and } \beta_v =
(1 + \epsilon)^j \beta_{min}\}$.
We refer to $\lset_j$ sets by \emph{level sets}.
We note that some of these sets may be empty.
There are two main sources of possible suboptimality in the fractional allocation that \cref{alg:local} finds:

\begin{itemize}
    \item Over-allocation: If $\alloc_v$ is greater than $\mcap_v$, an $\alloc_v - \mcap_v$ amount of fraction allocated to $v$ will not be counted towards the objective.
    \item Under-allocation: If $\alloc_v$ is less than $\mcap_v$, an extra capacity of $\mcap_v - \alloc_v$ is left to be exploited for vertex $v$.
\end{itemize}
The following lemma shows that for vertices not in $\lset_{2\tau}$, the under-allocation loss is negligible, and for vertices not in $\lset_0$, the over-allocation loss is negligible.

\begin{lemma} \label{lem:alloc-bound} Assume that we run \cref{alg:local} for $\tau \geq 1$ rounds. Then, after all rounds have completed, we have
\begin{enumerate}
    \item for any $v \in \bigcup_{j=0}^{2\tau-1} \lset_j$, $\alloc_v \geq \frac{1}{1+3\epsilon} \mcap_v$, and
    \item for any $v \in \bigcup_{j=1}^{2\tau} \lset_j$, $\alloc_v \leq (1+3\epsilon) \mcap_v$.
\end{enumerate}
\end{lemma}
\begin{proof} Due to the symmetry of the two claims, we only prove the former.
Since $v$ is not in level set $\lset_{2\tau}$, there was a time that we did not increase $\beta_v$.
Let $t$ be the last round where $\beta_v$ is not increased.
At the end of round $t$, $\alloc_v \geq \frac{1}{1 + \epsilon}\mcap_v$.
For $t = \tau$, this completes the proof.
Otherwise, we consider the round $t + 1 \leq \tau$.
We first show that $\alloc_v \geq \frac{1}{1+3\epsilon} \mcap_v$ at the end of round $t + 1$.
Recall that $\alloc_v = \sum_{u \in \neigh_v} \frac{\beta_v}{\sum_{a' \in \neigh_u} \beta_{a'}}$.
If $\beta_v$ is unchanged at round $t$, the numerator of each term also remains unchanged.
The denominator terms are increased at most by a factor of $(1 + \epsilon)$.
So in total, $\alloc_v$ is not decreased by more than a factor of $(1 + \epsilon)$.
Hence, $\alloc_v \geq \frac{1}{(1+\epsilon)^2} \mcap_v \geq \frac{1}{(1+3\epsilon)} \mcap_v $ at the end of round $t + 1$.
In the other case, $\beta_v$ is decreased at round $t$, so the numerator of each term is also reduced by a factor of $(1 + \epsilon)$.
In total, $\alloc_v$ is decreased by a factor of at most $(1+\epsilon)^2$ at round $t + 1$.
Note that the reduction of $\beta_v$ at round $t$ means $\alloc_v$ was at least $(1 + \epsilon) \mcap_v$, and therefore at least $\frac{1}{1+\epsilon}$ at round $t + 1$.
So independent of whether $\beta_v$ was reduced or not, $\alloc_v \geq \frac{1}{1+3\epsilon} \mcap_v$ holds at round $t + 1$.

By definition of $t$, $\beta_v$ is increased in all rounds after $t$.
With a similar argument, we know that $\alloc_v$ does not decrease at any of these rounds.
So $\alloc_v$ remains at least $\frac{1}{1+3\epsilon} \mcap_v$ till the last round.
\end{proof}

We remark that \cref{lem:alloc-bound} is an adaptation of \cite[Lemma 1]{AZM18}.

For ease of notation, for a subset $S$ of vertices in $R$, we use $\mcap(S) = \sum_{v \in S} \mcap_v$ to denote the total capacity of $S$, and use $\neigh(S) = \bigcup_{v \in S} \neigh_v$ to denote the neighbor set of $S$.
\cref{lem:alloc-bound} implies the two lower bounds on $\mweight$. First, since each vertex in $\bigcup_{j<2\tau} \lset_j$ almost fill up their capacities, we have

\begin{equation} \label{eqn:weight-bound-1} \textstyle
    \mweight \geq \frac{1}{1+3\epsilon} \sum_{j=0}^{2\tau-1} \mcap(\lset_j).
\end{equation}
Second, since all vertices in $\bigcup_{j > 0} \lset_j$ has bounded over-allocation, when line 6 of \cref{alg:local} was executed, the ratio of $\frac{\mcap_v}{\alloc_v}$ was at least $\frac{1}{1 + 3\epsilon}$.
Hence, we have

\begin{equation} \label{eqn:weight-bound-2} \textstyle
    \mweight \geq \frac{1}{1+3\epsilon} \sum_{j=1}^{2\tau} \sum_{v \in \lset_j} \alloc_v.
\end{equation}

Let $OPT$ be the weight of an optimal fractional allocation. The following lemma gives an upper bound on $OPT$.

\begin{lemma} \label{lem:goal} Assume that we run \cref{alg:local} for $\tau \geq 1$ rounds. Then, after all rounds have completed, we have
\begin{equation} \label{eqn:goal} \textstyle
     OPT \leq (1 + 3\epsilon) \mweight + |\neigh(\lset_{2\tau})|.
\end{equation}
\end{lemma} 
\begin{proof}
In \cite[Claim 1]{AZM18}, it has been shown that
\begin{equation*} \textstyle
    \mbox{for all } 0 \leq \ell \leq 2\tau \mbox{, we have } OPT \leq \sum_{j=0}^{\ell} \mcap(\lset_j) + |\neigh (\bigcup_{j=\ell+1}^{2\tau} \lset_{j})|.
\end{equation*}
That is, the optimal fractional allocation is at most the total capacity of $\bigcup_{j \leq \ell} \lset_j$ plus the total number of neighbors of $\bigcup_{j > \ell} \lset_j$.
By setting $\ell = 2\tau - 1$, we obtain
\begin{equation*} \textstyle
    OPT \leq \sum_{j=0}^{2\tau-1} \mcap(\lset_j) + |\neigh(\lset_{2\tau})|.
\end{equation*}
By \cref{eqn:weight-bound-1}, the term $\sum_{j=0}^{2\tau-1} \mcap(\lset_j)$ can be replaced by $(1 + 3\epsilon) \mweight$, which yields the desired inequality.
\end{proof}

\begin{remark}
Before proceeding, we remark on how the arboricity naturally arises in our analysis as follows.
By \cref{lem:goal}, if the current fractional solution is as large as $(1 - O(\eps)) |\neigh(\lset_{2\tau})|$, then the solution is a $(2 + O(\eps))$ approximation.
The only reason that $\mweight$ does not reach $(1 - O(\eps)) |\neigh(\lset_{2\tau})|$ is when most fractions from $\neigh(\lset_{2\tau})$ are allocated to $\lset_0$, because $\lset_0$ is the only set of vertices having unbounded over-allocation.
Intuitively, this implies that the subgraph induced by $\neigh(\lset_{2\tau}) \cup \lset_0$ is dense because the vertices in $\lset_0$ have low priority values compared to the vertices in $\lset_{2\tau}$.
For a large enough $\tau$, such a dense subgraph cannot exist in a low-arboricity graph, implying a contradiction.
\end{remark}


We are ready to prove the approximation factor of \cref{alg:local}.

\begin{theorem} \label{thm:apx-factor} Assume that we run \cref{alg:local} for $\tau \geq \loglambda$ rounds. Then, $OPT \leq (2 + 10\epsilon)\cdot \mweight$.
\end{theorem}
\begin{proof} Let $N'$ denote $\neigh(\lset_{2\tau})$. By \cref{eqn:goal}, if we can show that $|N'| \leq (1 + 7\epsilon) \cdot \mweight$, then $OPT \leq (2 + 10\epsilon)\mweight$ and thus \cref{alg:local} is a $(2 + 10\epsilon)$-approximation algorithm. To prove that $|N'| \leq (1 + 7\epsilon) \mweight$, we consider two cases:

\begin{itemize}
    \item[] \textbf{Case 1:} $|N'| \leq |\lset_0|$. By \cref{eqn:weight-bound-1}, we have
    \begin{equation*}
        \mweight \geq \frac{1}{1+3\epsilon} \mcap(\lset_0) \geq \frac{1}{1+3\epsilon} |\lset_0| \geq \frac{1}{1+3\epsilon}|N'|.
    \end{equation*}
    Hence, $|N'| \leq (1 + 7\epsilon) \mweight$ holds for this case.
    
    \item[] \textbf{Case 2:} $|N'| > |\lset_0|$. We first show that $\sum_{j=1}^{2\tau} \sum_{v \in \lset_j} \alloc_v \geq (1 - \frac{\epsilon}{2}) |N'|$.
    Consider a vertex $u \in L$ and a vertex $v \in \lset_0$.
    Since $N'$ is the neighbor set of $\lset_{2\tau}$, there is a vertex $v' \in \lset_{2\tau}$ that is a neighbor of $u$.
    In the end of round $\tau$, $\beta_{v'}$ is larger than $\beta_{v}$ by a factor of $(1 + \epsilon)^{2\tau}$.
    Thus, in the beginning of round $\tau$, $\beta_{v'}$ is larger than $\beta_v$ by at least a factor of $(1 + \epsilon)^{2\tau-2}$.
    Recall that $x_{u,v}$ and $x_{u,v'}$ are calculated proportional to the priority values of $v$ and $v'$ in the beginning of round $\tau$.
    Thus, $\vx_{u,v} \leq \vx_{u,v'} / (1+\epsilon)^{2\tau - 2}$.
    Since $\tau = \loglambda$, we have $\vx_{u,v} \leq \frac{\epsilon}{4\lambda}\cdot\vx_{u,v'} \leq \frac{\epsilon}{4\lambda}$.
    That is, every edge $(u, v)$ with $u \in N'$ and $v \in \lset_0$ allocates at most an amount of $\frac{\epsilon}{4\lambda}$ to $\lset_0$.
    Consider the subgraph $G'$ of $G$ induced by $N' \cup \lset_0$.
    Since the arboricity of $G$ is bounded by $\lambda$ and $|N'| > |\lset_0|$, the number of edges in $G'$ is at most $\lambda \cdot (|N'| + |\lset_0|) < 2 \lambda |N'|$.
    Therefore, the total allocation from $N'$ to $\lset_0$, $\sum_{u \in N'} \sum_{v \in \lset_0} \vx_{u,v}$, is at most $2 \lambda |N'| \cdot \frac{\epsilon}{4\lambda} = \frac{\epsilon}{2} |N'|$.
    That is, the total allocation from $N'$ to vertices not in $\lset_0$ is at least $(1 - \frac{\epsilon}{2})|N'|$.
    This implies that $\sum_{j=1}^{2\tau} \sum_{v \in \lset_j} \alloc_v \geq (1 - \frac{\epsilon}{2}) |N'|$.

    By \cref{eqn:weight-bound-2}, $\mweight$ is at least $\frac{1}{1+3\epsilon} \sum_{j=1}^{2\tau} \sum_{v \in \lset_j} \alloc_v$. Therefore, we obtain
    
    \begin{equation*}
        \mweight \geq \frac{1}{1+3\epsilon} \sum_{j=1}^{2\tau} \sum_{v \in \lset_j} \alloc_v \geq \frac{1}{1+3\epsilon} (1 - \frac{\epsilon}{2}) |N'| \geq \frac{1}{1 + 7\epsilon} |N'|,
    \end{equation*}
    where the last inequality holds for $\epsilon \leq 1$. Hence, $|N'| \leq (1 + 7\epsilon) \mweight$ also holds for this case.
\end{itemize}

Since we have $|N'| \leq (1 + 7\epsilon) \mweight$ in both cases, the theorem holds.
\end{proof}

\noindent Note that \cref{thm:LOCAL-allocation} is a direct consequence of \cref{thm:apx-factor}.

We conclude this section by the following remark. The proof of \cref{thm:apx-factor} also implies that after $\loglambda$ rounds, at least one of the following conditions must hold: either $\neigh(\lset_{2\tau}) \leq |\lset_0|$, or $\sum_{j=1}^{2\tau} \sum_{v \in \lset_j} \alloc_v \geq (1 - \frac{\epsilon}{2}) |N'|$; moreover, if one of the conditions hold at the end of the last round, the output is a $(2 + 10\epsilon)$-approximation. This implies a variation of \cref{alg:local} that does not assume the knowledge of $\lambda$: Instead of running the for-loop in lines 1-4 for $\loglambda$ rounds, we terminate the loop when one of the above conditions hold. By the discussion above, the variation is a $(2 + 10\epsilon)$-approximation algorithm that performs $\tau \leq \loglambda$ rounds of computation. The new termination condition for the loop can be checked in $O(1)$ MPC rounds. However, it is not clear how to check it in $O(1)$ LOCAL rounds.

\section{MPC algorithm}
\label{sec:mpc}

\begin{algorithm}
\begin{algorithmic}[1]
\Statex \textbf{Input:} $G = (L \cup R, E)$, capacity constraints $\{\mcap_v\}_{v \in R}$, number of rounds $\tau \geq 1$, approximation parameter $\epsilon > 0$.
\Statex \textbf{Output:} A $2+\epsilon$ approximate fractional allocation of $G$
\Statex \textbf{Parameters:} $B = B_{\epsilon/48}$ \text{ as per } \cref{eqn:bvalue}, and $t = (1+\epsilon)^{2B} \cdot \epsilon^{-5} \cdot \log n$
\Statex \textbf{Initialization:} Set $\beta_v \gets 1$ for every $v \in R$ and $\beta_u = \sum_{v\in \neigh_u} \beta_v$ for all $u \in L$

\For{ phases $p = 0, 1, \dots \tau/B-1$}
    \State for each $w\in V(G)$:
    \Statex \hspace{1cm} Partition $\neigh_w$ into $4\tau+1$ groups $\lset_{-2\tau},\dots \lset_{2\tau}$ as follows, 
    \Statex \hspace{3.5cm} $$\lset_x \gets \Big\{u \in \neigh_w \mid \beta_u \in ((1+\epsilon)^{x-1}, (1+\epsilon)^x] \Big\} \text{ for each } x = -2\tau, \cdots ,2\tau$$ \label{line:block-splitting}
    
    \State for each vertex $w \in V(G)$, integer $r \in [0, B)$, integer $x \in [-2\tau, 2\tau]$:
    \Statex \hspace{1cm} $E_{r, v, x} \gets $ Sample $t$ edges from $\neigh_v \bigcap \lset_x$ u.a.r.
    \Statex \hspace{0.25cm} $H_r:=$ Graph induced by edges in $\cup E_{r, v, x}$ over all $v, x$.
    \Statex \hspace{0.25cm} $H_{r, v}:=$ Ball of radius $B$ around $v$ in $H_r$.
    \Statex \hspace{0.25cm} $\neigh_{r, v}:=$ neighborhood of $v$ in $H_r$. \label{line:sampling}
    
    \For { rounds $r = pB+1, pB+2, \dots (p+1)B$} \label{line:local-sim-begin}
        \Statex \hspace{1cm} \textbf{for} each $u \in L$:
            \State \hspace{0.25cm} Set $\beta_{u} \gets \frac{|\neigh_u|}{|\neigh_{r, u}|} \sum\limits_{v \in \neigh_{r, u}} \beta_{v}$ \label{line:est1}
        \Statex \hspace{1cm} \textbf{for} each $v \in R$:
            \State \hspace{0.25cm} Set $\halloc_v \gets \frac{|\neigh_v|}{|\neigh_{r, v}|} \sum \limits_{u \in \neigh_{r, v}} \frac{\beta_v}{\beta_u}$ and update  $\beta_v$ as follows, \label{line:est2}
            
            \State \hspace{0.25cm} $$ \beta_v \gets \begin{cases} 
            \beta_v (1 + \epsilon) \ \ \ \text{if } \halloc_v \leq \mcap_v / (1 + \epsilon) \\
            \beta_v / (1 + \epsilon) \ \ \text{if } \halloc_v \geq \mcap_v (1 + \epsilon) \\
            \beta_v \ \ \ \ \ \ \ \ \ \ \ \ \text{otherwise}      
            \end{cases} $$
    \EndFor
    \label{line:local-sim-end}
    \State Each $v \in R$ calculates $\vx_{u, v} = \min(1, \frac{\mcap_v}{\halloc_v}) \cdot \sum\limits_{u \in \neigh_v} \frac{\beta_v}{\beta_u} $, the desired fractional allocation.
    \label{line:last-step}
\EndFor
\end{algorithmic}
\caption{A fractional allocation algorithm for bipartite graphs}
\label{alg:matching}
\end{algorithm}

In this section we outline how \cref{alg:local} can be efficiently implemented in sublinear MPC within $\tilde{O}(\sqrt{\log \lambda})$ rounds. Towards this goal, we present \cref{alg:matching}, a LOCAL algorithm with small \textit{locality volume}. The locality volume is the size of the subgraph around each vertex $v$ that is sufficient to obtain the output of $v$.  
More precisely, for every vertex $v \in V(G)$ and $\Theta(\sqrt{\log \lambda})$ rounds of \cref{alg:matching}, the output of $v$ is influenced by at most $O(\lambda \cdot \text{poly}(\log n))$ vertices. Mirroring the sparsification technique in \cite{GU19}, yields us our result. \cref{alg:local} gives a $2 + \epsilon$ approximation. Applying the framework of \cite{ghaffari2022massively}, we can boost this ratio to $1 + \epsilon$. A modification in the reduction is needed since the framework of \cite{ghaffari2022massively} is for the more general $b$-matching problem. We discuss details in \cref{sec:framework}.

\begin{theorem} \label{thm:fast-mpc}
    There is an $O_{\epsilon}(\sqrt{\log \lambda} \cdot \log \log \lambda)$-round MPC algorithm using $n^{\alpha}$ local memory (for any constant $\alpha \in (0, 1)$) and $\tilde{O}(\lambda n)$
global memory that with high probability computes an $2+\epsilon$ approximate allocation of any graph with arboricity $\lambda$. Using the framework of \cite{ghaffari2022massively}, the approximation can be further reduced to $1+\epsilon$ with the same asymptotic round complexity.
\end{theorem}

In \cref{alg:local}, in each round a vertex $v$ communicates with \textit{all} of its neighbors. The locality volume for $B$ rounds is thus, $\Delta^B$ in the worst case. The communication is done for exactly two purposes: to determine the value of either (i) $\sum\limits_{v \in \neigh_u} \beta_u$ for all $u \in L$ or (ii) $\sum\limits_{u \in \neigh_v} \beta^{-1}_u$ for all $v \in R$. 

For a node $v \in R$, $\beta_v$ has the same connotation as \cref{alg:local}. The values $\beta_v$ provide an implicit distribution of how weights should be assigned to the edges incident on it. For a node $u \in L$ we define $\beta_u = \sum\limits_{v \in \neigh_{u}} \beta_v$. This definition helps us treat the aggregations needed for Line 2 and Line 3 of \cref{alg:local} in a similar way. Note that to implement Line 2 we require $\sum\limits_{v \in \neigh_u} \beta_v$ for every vertex $u \in L$. As per our definition, Line 3 can be written as $\beta_v \cdot \sum\limits_{u \in \neigh_v} \beta^{-1}_{u}$, i.e., essentially $v$ requires only $\sum\limits_{u \in \neigh_v} \beta^{-1}_u$, which is symmetric to the previous aggregation.
Observe that $1 + O(\epsilon)$ multiplicative approximations of the $\beta$ values are sufficient to obtain $1 + O(\epsilon)$ approximations of both aggregates.

Finally at the end of a single round, if these aggregations exceed a certain threshold, the $\beta$ values are decreased by a $1+\epsilon$ factor. Similarly, if they fall below a certain threshold, the $\beta$ values are increased by a $1+\epsilon$ factor.

Our main observation is that even if these thresholds are slightly ``loose'', the algorithm is still correct. We prove this observation in \cref{sec:mpc:correctness}. A consequence of this is that estimating the exact aggregations computed by \cref{alg:local} in Lines 2 and 3 is sufficient.

To obtain approximations with an efficient number of samples, we use the fact that the $\beta$ values (see \cref{alg-local:agg1}) cannot change drastically within a single round. In particular, after $B$ rounds the values change by at most $(1+\epsilon)^B$. The following lemma helps in this regard.

\begin{lemma}\label{lem:sampling}
    Consider any sequence $X$ of $n$ elements $x_1, x_2, \dots x_n$ with $S_x = \sum\limits_{i=1}^n x_i$ and such that each $x_i$ lies in the interval $[V / t, V \cdot t]$ for some positive real $V$, and let $\epsilon \in (0, 1]$.
    For $s \geq 20 \cdot t^2 \log n / \epsilon^4$, let $y_1, y_2, \dots y_s$ be a sequence of uniformly random independent samples from $X$.
    Denote the rescaled sum of the samples as $S_y = \frac{n}{s} \cdot \sum\limits_{i=1}^s y_i$. Then,
    \begin{equation*}
        \mathrm{Pr}\Bigg[|S_y - S_x| \leq 4\epsilon S_x\Bigg] \geq 1 - n^{-10} \cdot \log_{1+\epsilon} t
    \end{equation*}
\end{lemma}

\begin{proof}
    We have $S_x \geq n V / t$. Let us split the interval $[V / t, V \cdot t]$ into groups $\lset_i = ((1+\epsilon)^{r-1}, (1 + \epsilon)^r]$ for $i$ ranging from $-\lceil \log_{1+\epsilon} t \rceil$ to $\lceil \log_{1+\epsilon} t \rceil$. Let $M_i$ be the set of elements $j$ whose value $x_j$ lies in $\lset_i$.

    We partition the contribution of the elements $x_i$ into two parts, depending on the size of $|M_i|$. Let $S_x = S_{x,\mathsf{small}} + S_{x,\mathsf{large}}$ where, $S_{x,\mathsf{small}}$ is the sum of all $x_j$ such that $j$ lies in a group $M_i$ of size at most $10 n \log n / (s \epsilon^2)$. $S_{x,\mathsf{large}}$ is the sum of the remaining elements.

    \begin{enumerate}
        \item[] \textbf{Case 1:} ($S_{x,\mathsf{small}} \leq S_x \epsilon$). In this case, the total contribution of such groups can be bound as follows,
        \begin{equation*}
            S_{x,\mathsf{small}} \leq \frac{20 n V t \log n}{s \epsilon^3} = \frac{n V}{t} \cdot \frac{20t^2 \log n}{s \epsilon^3} \leq S_x \cdot \epsilon
        \end{equation*}

        For the leftmost inequality, we assumed that every group has the maximum size \\ ($10 n \log n / (s \epsilon^2)$) and maximum value. The maximum value is bounded as follows, the value of the largest group is at most $Vt$. Therefore the total contribution over all groups is at most the infinite geometric sum with first term as $\frac{10 n V t \log n}{s \epsilon^2}$ and common ratio $1+\epsilon$, which is bounded by $\frac{20 n V t \log n}{s \epsilon^3}$. 
        For the last inequality, we used the assumption that $s \geq 20 t^2 \log n / \epsilon^4$.
        
        \item[] \textbf{Case 2:} ($|S_y - S_{x,\mathsf{large}}| \leq 3 \epsilon S_{x,\mathsf{large}}$). In this case, for each $x_j$ contributing to the sum, at least $10 n \log n / (s \epsilon^2)$ values close to it exist (i.e., they lie in the same set $M_i$ for some $i$). For each such group $M_i$, the expected number of samples $y_k$ chosen from $M_i$ is $10 \log n / \epsilon^2$. By Chernoff bounds, the number of samples is off by a factor of $1+2\epsilon$ of the expected value by a probability of at most $2 e^{-40 \log n/3} \leq n^{-10}$.

        The sum of the values of these samples incurs at most another $1+\epsilon$ factor, since the gap between the largest and smallest value in a group is $1+\epsilon$. The overall approximation ratio is $1 + 3\epsilon$ for this case.
    \end{enumerate}
    Assuming we lose all the elements in Case 1, we still have at least $(1-\epsilon) S_x$ left in Case 2. In Case 2, we do not deviate from the actual value by more than $3\epsilon S_x$. Overall, $S_y$ deviates from $S_x$ by at most $4\epsilon S_x$. Using the union bound over all intervals in Case 2 gives us the desired probability.
\end{proof}

Guided by \cref{lem:sampling}, if we want to compress $B$ rounds of \cref{alg:local} to MPC (being content with $1+O(\epsilon)$ approximations), then we can plugin $t = (1+\epsilon)^{B}$ in \cref{lem:sampling} and obtain $s$, the number of samples needed per group. The total degree per vertex is at most $d = 20(1+\epsilon)^{2B} \cdot \log^2 n \cdot \epsilon^{-5}$.

Now the $B$-radius neighborhood of the relevant vertices around each vertex has volume at most $d^B$. We choose $B$ so that (i) the neighborhood can be stored within a single machine and (ii) the size is at most $\lambda$. (i) ensures that we can simulate $B$ rounds for $v$ in a single machine and (ii) ensures that the total memory is $\tilde{O}(\lambda n + m)$. We require $d^B \leq \min\{\lambda, n^\alpha\}$:
\begin{equation*}
    2B^2 \log(1+\epsilon) + 5 B \log(\epsilon^{-1}) + 2B \log \log n + 5B \leq \min\{\log \lambda, \alpha \cdot \log n \}  
\end{equation*}
Since we are concerned only with asymptotics, we obtain a simple candidate for $B$ by bounding each term on the left by $\min(\log \lambda, \alpha \log n) / 4$. Assuming $\epsilon$ is a constant, and $n$ is large enough, the first term is dominant. Thus, we can set $B = B_{\epsilon}$ to guarantee that aggregations are computed within a $1+4\epsilon$ factor. 
\begin{equation}\label{eqn:bvalue}
B_{\epsilon} = \min \Bigg\{ \frac{\sqrt{\alpha \log n}}{\sqrt{8 \epsilon}}, \frac{\sqrt{\log \lambda}}{\sqrt{8 \epsilon}} \Bigg\}
\end{equation}

For ease of analysis in the correctness proof, we set $B=B_{\epsilon/48}$, giving us the following lemma.

\begin{lemma} \label{lem:approx}
    Consider \cref{alg:matching}. We define the following parameters, 
    \begin{equation*} \begin{split}
    \overline{\beta_u} = \sum\limits_{v \in \neigh_u} \beta_v, \ \ \ 
    \alloc_v = \sum\limits_{u \in \neigh_v} \overline{\beta_u}
    \end{split} \end{equation*}
    We have $|\overline{\beta_u} - \beta_u| \leq \epsilon/12 \cdot \overline{\beta_u}$ and $|\alloc_v - \halloc_v| \leq \epsilon/4 \cdot \alloc_v$ with probability at least $1 - n^{-5}$ over all rounds and vertices.
\end{lemma}

\begin{proof}
    For $\overline{\beta_u}$ we use \cref{lem:sampling} with $t = (1 + \epsilon)^B$. Observe that $B$ value is set so that the error term in \cref{lem:sampling} is $\epsilon / 48$. The approximation ratio is $1 + \epsilon/12$. Since the choice of $B$ in \cref{eqn:bvalue} is $o(n)$ even when $\epsilon < 1/n$, the probability guarantee is at least $1 - n^{-9}$. 

    For the second expression, we have the same situation. Applying union bound over all vertices and rounds we get a conservative failure probability of $n^{-5}$ overall. For the approximation ratio, the $1 + \epsilon/12$ is cascaded twice, leading to overall $(1 +  \epsilon/12)^2 \leq 1 + \epsilon / 4$.
\end{proof}
From \cref{lem:approx}, we can infer that the $B$ round simulation is identical to a run of the algorithm described in \cref{sec:mpc:correctness} with the value of $k$ chosen appropriately. Our round complexity is $\tau \log B / B$ where $\tau$ is given in \cref{alg:local} and $B$ is as per \cref{eqn:bvalue}, which completes the proof of \cref{thm:fast-mpc}.

We did not cover MPC simulations of certain lines in \cref{alg:matching} such as the graph exponentiation to compute the subgraphs $H_{r, v}$ for each $v$ and the random sampling from the groups. These can be implemented from standard primitives such as graph exponentiation and sorting, which are by now standard in the MPC literature.

Finally, our algorithms as stated assume knowledge of $\lambda$. To remove this assumption, we can make use of the fact that the proof of correctness of \cref{alg:local} provides a termination condition. To perform this test, we are required to compute the fractions assigned to edges incident to a subset of vertices. Again using known primitives in MPC, we can perform this computation in $O(1)$ rounds. We therefore can repeat the algorithm by guessing the value of $\sqrt{\log \lambda}$ and doubling the estimate when the termination test fails.

\section{From fractional to integral allocation}
\label{sec:rounding}

We describe in this section a simple rounding procedure that, in expectation, obtains an $\Theta(1)$-approximate integral allocation from a $4 + \epsilon$-approximate fractional allocation. In the MPC model, we can further obtain high probability of success by running $O(\log n)$ copies of this algorithm in parallel, and choosing the copy that corresponds to the largest rounded allocation.

Let $M_f$ be the given fractional allocation.
Denote by $\text{wt}(M_f)$ the total sum of the fractional values over all edges in $M_f$.
The algorithm is as follows.
First, sample edge $e$ independently of the others with probability $M_f(e) / 6$.
Next, if a vertex $v$ has the number of edges incident to it greater than its capacity, drop all the sampled edges incident to $v$. We call such vertex $v$ \emph{heavy}.

Let $X_e$ be a $0/1$ random variable equal to $1$ if $e = \{u, v\}$ is sampled, and $Y_e$ be a $0/1$ random variable equal to $1$ if $e$ remains among the sampled edges. Then, we have
\newcommand{\E}[1]{\mathbb{E}\left[ #1 \right]}
\newcommand{\prob}[1]{\Pr \left[ #1 \right]}
\begin{equation}\label{eq:E[Xe]}
    \E{Y_e} = \prob{X_e = 1} \cdot \prob{\text{neither $u$ nor $v$ is heavy} \mid X_e = 1}.
\end{equation}
Note that $\prob{X_e = 1} = \frac{M_f(e)}{6}$. We analyze $\prob{\text{$u$ is heavy} \mid X_e = 1}$ as follows. Let $E'$ be the set of edges $e' \neq e$ that are incident to $u$. Then,

\[
    \prob{\text{$u$ is heavy} \mid X_e = 1} = \prob{\sum_{e' \in E'} X_{e'} > \mcap_u - 1}.
\]
We show that $\prob{\sum_{e' \in E'} X_{e'} > \mcap_u - 1} \leq \frac{1}{3}$ by considering two cases:

\begin{enumerate}
    \item $C_u > 1$. In this case, $C_u - 1 \ge \frac{\mcap_u}{2}$. Therefore,
    \[
        \prob{\sum_{e' \in E'} X_{e'} > \mcap_u - 1} \le \prob{\sum_{e' \in E'} X_{e'} > \frac{\mcap_u}{2}}
    \]
    By our algorithm, $\mathbb{E}[\sum_{e' \in E'} X_{e'}] \leq \frac{\mcap_v}{6}$. Hence, by Markov's inequality, it implies
    \[
        \prob{\sum_{e' \in E'} X_{e'} > \mcap_u - 1} \le \frac{1}{3}.
    \]
    
    \item $C_u = 1$. In this case, we have
    \begin{align*}
        &\prob{\sum_{e' \in E'} X_{e'} > \mcap_u - 1} \\
        =& \prob{\text{at least one edge in $E'$ is chosen}} \leq \frac{1}{6},
    \end{align*}
    where the last inequality is obtained by applying the union bound.
\end{enumerate}

By the above analysis, we know that $\prob{\text{$u$ is heavy} \mid X_e = 1} \leq \frac{1}{3}$. By symmetry, we also have $\prob{\text{$v$ is heavy} \mid X_e = 1} \leq \frac{1}{3}$. Plugging this into \cref{eq:E[Xe]} yields
\[
    \E{Y_e} \ge \frac{M_f(e)}{3} \cdot \rb{1 - \frac{1}{3} - \frac{1}{3}} = \frac{M_f(e)}{9}.
\]

Let $M$ be the allocation kept after the rounding procedure. By the linearity of expectation, we have
\[
    \E{|M|} = \sum_{e \in E} \E{Y_e} \ge \text{wt}(M_f) / 9.
\]
\newcommand{\Mstar}{M^\star}
Let $\Mstar$ be the maximum allocation. Since $M_f$ is a $4 + \epsilon$ approximation, we have $\text{wt}(M_f) \ge |\Mstar| / 5$. Therefore,
\[
    \E{|M|} \ge \frac{|\Mstar|}{45}.
\]
Let $\alpha$ be the probability that $|M| \le |\Mstar| / 450$. Then, by using that $|M| \le |\Mstar|$ deterministically, we have
\[
    \alpha \frac{|\Mstar|}{450} + (1 - \alpha) |\Mstar|\ge \E{|M|} \ge \frac{|\Mstar|}{45}.
\]
A simple calculation leads to $\alpha \le 440 / 449$.
Therefore, with a constant positive probability, $|M|$ is a $\Theta(1)$ approximation.

\bibliographystyle{alpha}
\bibliography{references}

\newcommand{\etalchar}[1]{$^{#1}$}
\begin{thebibliography}{WWWC15}

\bibitem[ABB{\etalchar{+}}19]{assadi2019coresets}
Sepehr Assadi, MohammadHossein Bateni, Aaron Bernstein, Vahab Mirrokni, and Cliff Stein.
\newblock Coresets meet edcs: algorithms for matching and vertex cover on massive graphs.
\newblock In {\em Proceedings of the Thirtieth Annual ACM-SIAM Symposium on Discrete Algorithms}, pages 1616--1635. SIAM, 2019.

\bibitem[ALPZ21]{ahmadian2021distributed}
Sara Ahmadian, Allen Liu, Binghui Peng, and Morteza Zadimoghaddam.
\newblock Distributed load balancing: A new framework and improved guarantees.
\newblock In {\em 12th Innovations in Theoretical Computer Science Conference (ITCS 2021)}. Schloss Dagstuhl-Leibniz-Zentrum f{\"u}r Informatik, 2021.

\bibitem[AS22]{matching-paper2}
Susanne Albers and Sebastian Schubert.
\newblock Online ad allocation in bounded-degree graphs.
\newblock In Kristoffer~Arnsfelt Hansen, Tracy~Xiao Liu, and Azarakhsh Malekian, editors, {\em Web and Internet Economics - 18th International Conference, {WINE} 2022, Troy, NY, USA, December 12-15, 2022, Proceedings}, volume 13778 of {\em Lecture Notes in Computer Science}, pages 60--77. Springer, 2022.

\bibitem[AZM18]{AZM18}
Shipra Agrawal, Morteza Zadimoghaddam, and Vahab~S. Mirrokni.
\newblock Proportional allocation: Simple, distributed, and diverse matching with high entropy.
\newblock In Jennifer~G. Dy and Andreas Krause, editors, {\em Proceedings of the 35th International Conference on Machine Learning, {ICML} 2018, Stockholmsm{\"{a}}ssan, Stockholm, Sweden, July 10-15, 2018}, volume~80 of {\em Proceedings of Machine Learning Research}, pages 99--108. {PMLR}, 2018.

\bibitem[BBD{\etalchar{+}}19]{BBD+19}
Soheil Behnezhad, Sebastian Brandt, Mahsa Derakhshan, Manuela Fischer, MohammadTaghi Hajiaghayi, Richard~M. Karp, and Jara Uitto.
\newblock Massively parallel computation of matching and mis in sparse graphs.
\newblock In {\em Proceedings of the 2019 ACM Symposium on Principles of Distributed Computing}, PODC '19, page 481–490, New York, NY, USA, 2019. Association for Computing Machinery.

\bibitem[BHH19]{behnezhad2019exponentially}
Soheil Behnezhad, Mohammad~Taghi Hajiaghayi, and David~G Harris.
\newblock Exponentially faster massively parallel maximal matching.
\newblock In {\em 2019 IEEE 60th Annual Symposium on Foundations of Computer Science (FOCS)}, pages 1637--1649. IEEE, 2019.

\bibitem[BLM23]{matching-paper6}
Santiago~R. Balseiro, Haihao Lu, and Vahab Mirrokni.
\newblock The best of many worlds: Dual mirror descent for online allocation problems.
\newblock {\em Oper. Res.}, 71(1):101--119, 2023.

\bibitem[CDP21]{CDP21}
Artur Czumaj, Peter Davies, and Merav Parter.
\newblock Component stability in low-space massively parallel computation.
\newblock In {\em Proceedings of the 2021 ACM Symposium on Principles of Distributed Computing}, PODC'21, page 481–491, New York, NY, USA, 2021. Association for Computing Machinery.

\bibitem[C{\L}M{\etalchar{+}}18]{czumaj2018round}
Artur Czumaj, Jakub {\L}{\k{a}}cki, Aleksander M{\k{a}}dry, Slobodan Mitrovi{\'c}, Krzysztof Onak, and Piotr Sankowski.
\newblock Round compression for parallel matching algorithms.
\newblock In {\em Proceedings of the 50th Annual ACM SIGACT Symposium on Theory of Computing}, pages 471--484, 2018.

\bibitem[DD{\L}M24]{dhulipala2024parallel}
Laxman Dhulipala, Michael Dinitz, Jakub {\L}{\k{a}}cki, and Slobodan Mitrovi{\'c}.
\newblock Parallel set cover and hypergraph matching via uniform random sampling.
\newblock In {\em 38th International Symposium on Distributed Computing (DISC 2024)}, pages 19--1. Schloss Dagstuhl--Leibniz-Zentrum f{\"u}r Informatik, 2024.

\bibitem[DG08]{dean2008mapreduce}
Jeffrey Dean and Sanjay Ghemawat.
\newblock Mapreduce: simplified data processing on large clusters.
\newblock {\em Communications of the ACM}, 51(1):107--113, 2008.

\bibitem[DJSW19]{matching-paper4}
Nikhil~R. Devanur, Kamal Jain, Balasubramanian Sivan, and Christopher~A. Wilkens.
\newblock Near optimal online algorithms and fast approximation algorithms for resource allocation problems.
\newblock {\em J. {ACM}}, 66(1):7:1--7:41, 2019.

\bibitem[DSSX21]{matching-paper1}
John~P. Dickerson, Karthik~A. Sankararaman, Aravind Srinivasan, and Pan Xu.
\newblock Allocation problems in ride-sharing platforms: Online matching with offline reusable resources.
\newblock {\em ACM Trans. Econ. Comput.}, 9(3), June 2021.

\bibitem[FKM{\etalchar{+}}09]{feldman2009online}
Jon Feldman, Nitish Korula, Vahab Mirrokni, Shanmugavelayutham Muthukrishnan, and Martin P{\'a}l.
\newblock Online ad assignment with free disposal.
\newblock In {\em International workshop on internet and network economics}, pages 374--385. Springer, 2009.

\bibitem[GGJ20]{GGJ20}
Mohsen Ghaffari, Christoph Grunau, and Ce~Jin.
\newblock {Improved MPC Algorithms for MIS, Matching, and Coloring on Trees and Beyond}.
\newblock In Hagit Attiya, editor, {\em 34th International Symposium on Distributed Computing (DISC 2020)}, volume 179 of {\em Leibniz International Proceedings in Informatics (LIPIcs)}, pages 34:1--34:18, Dagstuhl, Germany, 2020. Schloss Dagstuhl -- Leibniz-Zentrum f{\"u}r Informatik.

\bibitem[GGK{\etalchar{+}}18]{GGK+18}
Mohsen Ghaffari, Themis Gouleakis, Christian Konrad, Slobodan Mitrovi\'{c}, and Ronitt Rubinfeld.
\newblock Improved massively parallel computation algorithms for mis, matching, and vertex cover.
\newblock In {\em Proceedings of the 2018 ACM Symposium on Principles of Distributed Computing}, PODC '18, page 129–138, New York, NY, USA, 2018. Association for Computing Machinery.

\bibitem[GGM22]{ghaffari2022massively}
Mohsen Ghaffari, Christoph Grunau, and Slobodan Mitrovi{\'c}.
\newblock Massively parallel algorithms for b-matching.
\newblock In {\em Proceedings of the 34th ACM Symposium on Parallelism in Algorithms and Architectures}, pages 35--44, 2022.

\bibitem[GKU19]{GKU19}
Mohsen Ghaffari, Fabian Kuhn, and Jara Uitto.
\newblock Conditional hardness results for massively parallel computation from distributed lower bounds.
\newblock In David Zuckerman, editor, {\em 60th {IEEE} Annual Symposium on Foundations of Computer Science, {FOCS} 2019, Baltimore, Maryland, USA, November 9-12, 2019}, pages 1650--1663. {IEEE} Computer Society, 2019.

\bibitem[GSZ11]{goodrich2011sorting}
Michael~T Goodrich, Nodari Sitchinava, and Qin Zhang.
\newblock Sorting, searching, and simulation in the mapreduce framework.
\newblock In {\em International Symposium on Algorithms and Computation}, pages 374--383. Springer, 2011.

\bibitem[GU19]{GU19}
Mohsen Ghaffari and Jara Uitto.
\newblock Sparsifying distributed algorithms with ramifications in massively parallel computation and centralized local computation.
\newblock In Timothy~M. Chan, editor, {\em Proceedings of the Thirtieth Annual {ACM-SIAM} Symposium on Discrete Algorithms, {SODA} 2019, San Diego, California, USA, January 6-9, 2019}, pages 1636--1653. {SIAM}, 2019.

\bibitem[HMZ11]{matching-paper3}
Bernhard Haeupler, Vahab~S. Mirrokni, and Morteza Zadimoghaddam.
\newblock Online stochastic weighted matching: Improved approximation algorithms.
\newblock In Ning Chen, Edith Elkind, and Elias Koutsoupias, editors, {\em Internet and Network Economics - 7th International Workshop, {WINE} 2011, Singapore, December 11-14, 2011. Proceedings}, volume 7090 of {\em Lecture Notes in Computer Science}, pages 170--181. Springer, 2011.

\bibitem[KSV10]{karloff2010model}
Howard Karloff, Siddharth Suri, and Sergei Vassilvitskii.
\newblock A model of computation for mapreduce.
\newblock In {\em Proceedings of the twenty-first annual ACM-SIAM symposium on Discrete Algorithms}, pages 938--948. SIAM, 2010.

\bibitem[LKK23]{liu2023scalable}
Quanquan~C Liu, Yiduo Ke, and Samir Khuller.
\newblock Scalable auction algorithms for bipartite maximum matching problems.
\newblock {\em arXiv preprint arXiv:2307.08979}, 2023.

\bibitem[LMSV11]{lattanzi2011filtering}
Silvio Lattanzi, Benjamin Moseley, Siddharth Suri, and Sergei Vassilvitskii.
\newblock Filtering: a method for solving graph problems in mapreduce.
\newblock In {\em Proceedings of the twenty-third annual ACM symposium on Parallelism in algorithms and architectures}, pages 85--94, 2011.

\bibitem[LW10]{lenzen2010brief}
Christoph Lenzen and Roger Wattenhofer.
\newblock Brief announcement: Exponential speed-up of local algorithms using non-local communication.
\newblock In {\em Proceedings of the 29th ACM SIGACT-SIGOPS symposium on Principles of distributed computing}, pages 295--296, 2010.

\bibitem[MSVV07]{mehta2007adwords}
Aranyak Mehta, Amin Saberi, Umesh Vazirani, and Vijay Vazirani.
\newblock Adwords and generalized online matching.
\newblock {\em Journal of the ACM (JACM)}, 54(5):22--es, 2007.

\bibitem[VVS10]{vee2010optimal}
Erik Vee, Sergei Vassilvitskii, and Jayavel Shanmugasundaram.
\newblock Optimal online assignment with forecasts.
\newblock In {\em Proceedings of the 11th ACM conference on Electronic commerce}, pages 109--118, 2010.

\bibitem[WWWC15]{wang2015socially}
Li~Wang, Huaqing Wu, Wei Wang, and Kwang-Cheng Chen.
\newblock Socially enabled wireless networks: Resource allocation via bipartite graph matching.
\newblock {\em IEEE Communications Magazine}, 53(10):128--135, 2015.

\bibitem[ZWMS20]{matching-paper5}
Goran Zuzic, Di~Wang, Aranyak Mehta, and D.~Sivakumar.
\newblock Learning robust algorithms for online allocation problems using adversarial training.
\newblock {\em CoRR}, abs/2010.08418, 2020.

\end{thebibliography}

\onecolumn
\appendix
\section{Correctness of the MPC algorithm}
\label{sec:mpc:correctness}

\begin{algorithm}
\begin{algorithmic}[1]
\Statex \textbf{Input:} $G = (L \cup R, E)$, capacity constraints $\{\mcap_v\}_{v \in R}$, number of rounds $\tau \geq 1$, approximation parameter $\epsilon \in (0, \frac{1}{4}]$.
\Statex \textbf{Parameters:} positive real numbers $\{k_{v,r}\}_{v \in R, 1 \leq r \leq \tau}$.
\Statex \textbf{Output:} a fractional allocation $\{\vx'_{u,v}\}_{(u,v) \in E}$
\Statex \textbf{Initialization:} Set $\beta_v \gets 1$ for every $v \in R$

\For{rounds $r = 1, 2, \dots \tau$}
    \State for each vertex $u \in L$: 
    \Statex \hspace{1cm} Set $\vx_{u, v} = \frac{\beta_v}{\sum_{v' \in \neigh_u} \beta_{v'}}$ for each $v \in \neigh_u$ 
    \State for each vertex $v \in R$: 
    \Statex \hspace{1cm} Set $\alloc_v \gets \sum\limits_{u \in \neigh_v} \vx_{u, v}$
    \State for each vertex $v \in R$, update  $\beta_v$ as follows, 
            
    \Statex  $$ \beta_v \gets \begin{cases} 
        \beta_v (1 + \epsilon) \ \ \ \text{if } \alloc_v \leq \mcap_v / (1 + k_{v, r}\epsilon) \\
        \beta_v / (1 + \epsilon) \ \ \text{if } \alloc_v \geq \mcap_v (1 + k_{v, r} \epsilon) \\
        \beta_v \ \ \ \ \ \ \ \ \ \ \ \ \text{otherwise}      
    \end{cases} $$    
\EndFor

\For {each vertex $v \in R$}
    \State For all $u \in \neigh_v$, set $\vx'_{u,v} \gets \begin{cases} 
        \frac{\mcap_v}{\alloc_v}\vx_{u,v} \ \ \ \text{if } \alloc_v > \mcap_v \\
        \vx_{u,v} \ \ \ \ \ \ \ \ \ \ \ \ \text{otherwise}      
    \end{cases}$
\EndFor
\end{algorithmic}
\caption{An adapted version of \cref{alg:local}.}
\label{alg:det-thresh}
\end{algorithm}

To prove the correctness of the MPC algorithm (\cref{alg:matching}), we consider an adaptation of our LOCAL algorithm (\cref{alg:local}).
The adaptation is given in \cref{alg:det-thresh}.
\cref{alg:local,alg:det-thresh} only have two differences:
First, \cref{alg:det-thresh} requires an additional set of parameters $\{k_{v,r}\}_{v \in R, 1 \leq r \leq \tau}$.
Second, \cref{alg:det-thresh} applies a different update rule in line 4, in which the thresholds for updating $\beta_v$ are replaced with $\mcap_v / (1 + k_{v,r}\epsilon)$ and $\mcap_v(1 + k_{v,r}\epsilon)$.
This adapted update rule is more flexible because the thresholds can be controled by adjusting the parameters.
Note that \cref{alg:local} is a special case where all parameters $k_{v,r}$ are fixed as 1.

Consider an input instance $(G, \{\mcap_v\}_{v \in R}, \tau, \epsilon)$ of \cref{alg:matching}.
The correctness proof for \cref{alg:matching} consists of three parts.

\begin{itemize}
    \item In the first part, we argue that \cref{alg:matching} is equivalent to \cref{alg:det-thresh} with a specific set of parameters.
More precisely, we show that given a fixed execution of \cref{alg:matching} (on the input instance $(G, \{\mcap_v\}_{v \in R}, \tau, \epsilon)$),
with high probability, there exists a set of parameters $\{k_{v,r}\}_{v \in R, 1 \leq r \leq \tau}$ such that the outputs of \cref{alg:matching,alg:det-thresh} are the same; in addition, the constructed parameters satisfy $\frac{1}{4} \leq k_{v,r} \leq 4$ for all $v \in R$ and $1 \leq r \leq \tau$.
    \item In the second part, we show that \cref{alg:det-thresh} has an approximation factor of $(2 + (2k + 8)\epsilon)$ if the input and parameters satisfy $\frac{1}{k} \leq k_{v,r} \leq k$ for some small number $k$, $\epsilon \leq \frac{1}{k}$, and $\tau \geq \loglambda$. By combining the results of the first and second parts, we know that \cref{alg:matching} has an approximation factor of $2 + 16 \epsilon$ if $\tau \geq \loglambda$ and $\epsilon \leq \frac{1}{4}$.
    \item In the third part, we show that \cref{alg:det-thresh} has an approximation factor of $1 + (k + 14)\epsilon$ if the input and parameters satisfy $\frac{1}{k} \leq k_{v,r} \leq k$ for some small number $k$, $\epsilon \leq \frac{1}{k}$, and $\tau \geq \logn$. By combining the results of the first and third parts, we know that \cref{alg:matching} has an approximation factor of $1 + 18\epsilon$ for $\tau = O(\logn)$ and $\epsilon \leq \frac{1}{4}$.
\end{itemize}
The three parts are presented in \cref{sec:appendix-1,sec:appendix-2,sec:appendix-3}, respectively.

\subsection{Equivalence of \texorpdfstring{\cref{alg:matching,alg:det-thresh}}{Algorithms 2 and 3}} \label{sec:appendix-1}
In the following, we argue that \cref{alg:matching} is equivalent to \cref{alg:det-thresh} with a specific set of parameters. Consider an execution of \cref{alg:matching}. For ease of presentation, for a fixed round, we use $\{\hbeta_v\}_{v \in R}$ to denote the priority values maintained by \cref{alg:matching}, and use $\{\beta_v\}_{v \in R}$ to denote the priority values maintained by \cref{alg:det-thresh}. Let $\hvx_{u, v} = \frac{\hbeta_v}{\sum_{v' \in \neigh_u} \hbeta_{u'}}$ denote the corresponding fractional values in \cref{alg:matching}.

\begin{lemma} \label{lem:equivalence} Let $(G, \{\mcap_v\}_{v \in A}, \epsilon, R)$ be an input instance for \cref{alg:matching,alg:det-thresh}.
Consider an execution of \cref{alg:matching} on the instance.
With high probability, there exists a set of parameters $\{k_{v,r}\}_{v \in R, 1 \leq r \leq \tau}$ for \cref{alg:det-thresh},
all in the range $[\frac{1}{4}, 4]$,
such that the output of the execution is the same as the output of \cref{alg:det-thresh}.
\end{lemma}
\begin{proof} We prove this lemma by constructing a set of parameters $\{k_{v,r}\}_{v \in R, 1 \leq r \leq \tau}$ such that $\beta_v = \hbeta_v$ holds for all $v \in R$ across all rounds.
Initially, both algorithms set the priority values as 1.
Hence, in beginning of the the first round, $\beta_v = \hbeta_v$ holds for all $v \in R$.

Consider a fixed round $r$ and assume that $\beta_v = \hbeta_v$ holds in the beginning of round $r$.
Let $v \in R$ be a vertex.
In the following, we show how to determine the parameter $k_{v,r}$ such that the two algorithms update the priority values in the same way.
Since $\beta_v = \hbeta_v$ holds for all $v \in R$, we know that $\vx_{u,v} = \hvx_{u,v}$ holds for each edge $(u, v)$ in round $r$.
Therefore, $\alloc_v = \sum_{u \in \neigh_v} \vx_{u,v}$ is equal to $\sum_{u \in \neigh_v} \hvx_v$.
By \cref{lem:approx}, the execution of \cref{alg:matching} satisfies
\begin{equation} \label{eqn:appendix-mpc-approx}
    |\halloc_v - \sum_{u \in \neigh_v} \hvx_v| = |\halloc_v - \alloc_v| \leq \frac{\epsilon}{4} \alloc_v
\end{equation}
with high probability.
Assume that this inequality holds for the given execution.
Consider two cases:
\begin{enumerate}
    \item[] \textbf{Case 1:} In round $r$, $\alloc_v \leq 2\mcap_v$. By \cref{eqn:appendix-mpc-approx}, we have
    
    \begin{equation} \label{eqn:appendix-case1}
        |\alloc_v - \halloc_v| \leq \frac{\epsilon}{4} \alloc_v \leq \frac{\epsilon}{2} \epsilon \mcap_v.
    \end{equation}
    There are three subcases:

    \begin{enumerate}
        \item[] \textbf{Case 1.1:} In the execution, $\hbeta_v$ is increased in round $r$.
        That is, $\halloc_v \leq \mcap_v / (1 + \epsilon)$.
        Our goal is to choose $k_{v,r}$ such that \cref{alg:det-thresh} also increases $\beta_v$.
        By \cref{eqn:appendix-case1}, we have $\alloc_v \leq \halloc_v + \frac{\epsilon}{2} \mcap_v \leq (\frac{1}{1 + \epsilon} + \frac{\epsilon}{2}) \mcap_v$.
        Thus, we should choose $k_{v,r}$ such that $\alloc_v \leq \frac{1}{(1 + k_{v,r} \epsilon)} \mcap_v$.
        Since $(\frac{1}{1 + \epsilon} + \frac{\epsilon}{2}) \leq \frac{1}{1 + \frac{1}{4}\epsilon}$ for $\epsilon \leq \frac{1}{4}$, we may choose $k_{v,r} = \frac{1}{4}$.
        
        \item[] \textbf{Case 1.2:} In the execution, $\hbeta_v$ is decreased in round $r$.
        That is, $\halloc_v \geq \mcap_v (1 + \epsilon)$.
        Our goal is to choose $k_{v,r}$ such that \cref{alg:det-thresh} also decreases $\beta_v$.
        By \cref{eqn:appendix-case1}, we have $\alloc_v \geq \halloc_v - \frac{\epsilon}{2} \mcap_v \geq (1 + \epsilon - \frac{\epsilon}{2}) \mcap_v$.
        By choosing $k_{v,r} = \frac{1}{2}$, we ensure that $\alloc_v \geq (1 + k_{v,r} \epsilon) \mcap_v$ and thus \cref{alg:det-thresh} will decrease $\beta_v$.
        
        \item[] \textbf{Case 1.3:} In the execution, $\hbeta_v$ is unchanged in round $r$.
        That is, $\mcap_v / (1 + \epsilon) < \halloc_v < \mcap_v (1 + \epsilon)$.
        By \cref{eqn:appendix-case1}, we have $(\frac{1}{1 + \epsilon} - \frac{\epsilon}{2}) \mcap_v < \alloc_v < (1 + \frac{3\epsilon}{2}) \mcap_v$.
        Using a similar argument, we know that \cref{alg:det-thresh} will not change $\beta_v$ if $\frac{1}{1 + k_{v,r}\epsilon} \leq \frac{1}{1 + \epsilon} - \frac{\epsilon}{2}$ and $1 + k_{v,r}\epsilon \geq 1 + \frac{3\epsilon}{2}$.
        It is not hard to verify that $k_{v,r} = 3$ satisfies these conditions.
    \end{enumerate}
    
    \item[] \textbf{Case 2:} In round $r$, $\alloc_v > 2\mcap_v$.
    In this case, we choose $k_{v,r} = 1$ so that \cref{alg:det-thresh} will decrease $\beta_v$.
    By \cref{eqn:appendix-case1}, we have
    \begin{equation*}
        \halloc_v \geq (1 - \epsilon) \alloc_v
                  > 2(1 - \epsilon)\mcap_v > (1 + \epsilon)\mcap_v,
    \end{equation*}    
    where the last inequality holds because $\epsilon \leq \frac{1}{4}$.
    Therefore, \cref{alg:matching} must also have decreased $\hbeta_v$ in round $r$.
\end{enumerate}

By repeating the above argument for $r = 0, 1, \dots, R$, we obtain a set of parameters $\{k_{v,r}\}$ such that, in each round, the two algorithms update the priority value of each vertex in exactly the same way. Hence, the lemma holds.
\end{proof}

\subsection{Performance guarantee on low arboricity bipartite graphs} \label{sec:appendix-2}

In the following, we show that \cref{alg:det-thresh} has an approximation factor of $2 + (2k + 8)\epsilon$ if the parameters satisfy $\tau \geq \loglambda$, $\frac{1}{k} \leq k_{v,r} \leq k$ holds for all $k_{v,r}$ and for some number $k$, and $\epsilon \leq \frac{1}{k}$.
By \cref{lem:equivalence}, this result implies that \cref{alg:matching} has an approximation factor of $2 + 16 \epsilon$ if $\tau \geq \loglambda$ and $\epsilon \leq \frac{1}{4}$.
The result is proven by modifying the proof in \cref{sec:local}.

We first present the approximation factor of \cref{alg:det-thresh}. Let $k$ be a number such that $\frac{1}{k} \leq k_{v,r} \leq k$ holds for all $v \in R$ and $1 \leq r \leq \tau$.
As in \cref{sec:local}, we consider the $\beta$ variables after the end of round $\tau$ and partition $R$ into level sets $\lset_0, \lset_1, \dots, \lset_{2\tau}$, where $\lset_i = \{v \mid v \in R \mbox{ and } \beta_v = (1 + \epsilon)^{i}\}$.
We generalize \cref{lem:alloc-bound} as follows.

\begin{lemma} \label{lem:appendix-alloc-bound} Assume that we run \cref{alg:det-thresh} for $\tau \geq 1$ rounds. Then, after all rounds have completed, we have
\begin{enumerate}
    \item for any $v \in \bigcup_{j=0}^{2\tau-1} \lset_j$, $\alloc_v \geq \frac{1}{1+(k + 2)\epsilon} \mcap_v$, and
    \item for any $v \in \bigcup_{j=1}^{2\tau} \lset_j$, $\alloc_v \leq (1+(k + 2)\epsilon) \mcap_v$.
\end{enumerate}
\end{lemma}
\begin{proof} Due to the symmetry of the two claims, we only prove the former.
Since $v$ is not in level set $\lset_{2\tau}$, there was a time that we did not increase $\beta_v$.
Let $t$ be the last round where $\beta_v$ is not increased.
At the end of round $t$, $\alloc_v \geq \frac{1}{1 + k\epsilon}\mcap_v$.
For $t = \tau$, this completes the proof.
Otherwise, we consider the round $t + 1 \leq \tau$.
We first show that $\alloc_v \geq \frac{1}{1+(k+2)\epsilon} \mcap_v$ at the end of round $t + 1$.
Recall that $\alloc_v = \sum_{u \in \neigh_v} \frac{\beta_v}{\sum_{a' \in \neigh_u} \beta_{a'}}$.
If $\beta_v$ is unchanged at round $t$, the numerator of each term also remains unchanged.
The denominator terms are increased at most by a factor of $(1 + \epsilon)$.
So in total, $\alloc_v$ is not decreased by more than a factor of $(1 + \epsilon)$.
Hence, at the end of round $t + 1$, we have $\alloc_v \geq \frac{1}{(1+k\epsilon)(1+\epsilon)} \mcap_v \geq \frac{1}{(1+(k+2)\epsilon)} \mcap_v $, where the last inequality holds because $k\epsilon \leq 1$.
In the other case, $\beta_v$ is decreased at round $t$, so the numerator of each term is also reduced by a factor of $(1 + \epsilon)$.
In total, $\alloc_v$ is decreased by a factor of at most $(1+\epsilon)^2$ at round $t + 1$.
Note that the reduction of $\beta_v$ at round $t$ means $\alloc_v$ was at least $(1 + \frac{1}{k}\epsilon) \mcap_v \geq \mcap_v$, and therefore at least $\frac{1}{(1+\epsilon)^2} \leq \frac{1}{(1+(k+2)\epsilon)}$ at round $t + 1$.
So independent of whether $\beta_v$ was reduced or not, $\alloc_v \geq \frac{1}{1+(k+2)\epsilon} \mcap_v$ holds at round $t + 1$.

By definition of $t$, $\beta_v$ is increased in all rounds after $t$.
With a similar argument, we know that $\alloc_v$ does not decrease at any of these rounds.
So $\alloc_v$ remains at least $\frac{1}{1+(k+2)\epsilon} \mcap_v$ till the last round.
\end{proof}

\cref{lem:appendix-alloc-bound} implies the two lower bounds on $\mweight$. First, since each vertex in $\bigcup_{j<2\tau} \lset_j$ almost fill up their capacities, we have

\begin{equation} \label{eqn:appendix-weight-bound-1} \textstyle
    \mweight \geq \frac{1}{1+(k+2)\epsilon} \sum_{j=0}^{2\tau-1} \mcap(\lset_j).
\end{equation}
Second, since all vertices in $\bigcup_{j > 0} \lset_j$ has bounded over-allocation, when line 6 of \cref{alg:det-thresh} was executed, the ratio of $\frac{\mcap_v}{\alloc_v}$ was at least $\frac{1}{1 + (k+2)\epsilon}$.
Hence, we have

\begin{equation} \label{eqn:appendix-weight-bound-2} \textstyle
    \mweight \geq \frac{1}{1+(k+2)\epsilon} \sum_{j=1}^{2\tau} \sum_{v \in \lset_j} \alloc_v.
\end{equation}

In the proof of \cref{lem:goal}, we have mentioned the following upper bounds on $OPT$, which is proven by \cite{AZM18}.

\begin{equation} \label{eqn:appendix-opt-bound} \textstyle
    OPT \leq \sum_{j=0}^{\ell} \mcap(\lset_j) + |\neigh (\bigcup_{j=\ell+1}^{2\tau} \lset_{j})| \mbox{ for all } 0 \leq \ell \leq 2\tau.
\end{equation}

We now present a generalization of \cref{lem:goal}.

\begin{lemma} Assume that we run \cref{alg:det-thresh} for $\tau \geq 1$ rounds. Then, after all rounds have completed, we have
\begin{equation} \label{eqn:appendix-goal} \textstyle
     OPT \leq (1 + (k+2)\epsilon) \mweight + |\neigh(\lset_{2\tau})|.
\end{equation}
\end{lemma}
\begin{proof}
By setting $\ell = 2\tau - 1$ in \cref{eqn:appendix-opt-bound}, we obtain $OPT \leq \sum_{j=0}^{2\tau-1} \mcap(\lset_j) + |\neigh(\lset_{2\tau})|$.
By \cref{eqn:appendix-weight-bound-1}, the term $\sum_{j=0}^{2\tau-1} \mcap(\lset_j)$ can be replaced by $(1 + (k+2)\epsilon) \mweight$, which yields the desired inequality.
\end{proof}

We are ready to prove the approximation factor of \cref{alg:det-thresh}.

\begin{theorem} \label{thm:appendix-apx-factor-1} Assume that we run \cref{alg:det-thresh} for $\tau \geq \loglambda$ rounds. Then, $OPT \leq (2 + (2k + 8)\epsilon)\mweight$.
\end{theorem}
\begin{proof} Let $N'$ denote $\neigh(\lset_{2\tau})$. By \cref{eqn:appendix-goal}, if we can show that $|N'| \leq (1 + (k + 6)\epsilon) \mweight$, then $OPT \leq (2 + (2k+8)\epsilon)\mweight$. To prove that $|N'| \leq (1 + (k+6)\epsilon) \mweight$, we consider two cases:

\begin{itemize}
    \item[] \textbf{Case 1:} $|N'| \leq |\lset_0|$. By \cref{eqn:appendix-weight-bound-1}, we have
    \begin{equation*}
        \mweight \geq \frac{1}{1+(k+2)\epsilon} \mcap(\lset_0) \geq \frac{1}{1+(k+2)\epsilon} |\lset_0| \geq \frac{1}{1+(k+2)\epsilon}|N'|.
    \end{equation*}
    Hence, $|N'| \leq (1 + (k+6)\epsilon) \mweight$ holds for this case.
    
    \item[] \textbf{Case 2:} $|N'| > |\lset_0|$. We first show that $\sum_{j=1}^{2\tau} \sum_{v \in \lset_j} \alloc_v \geq (1 - \frac{\epsilon}{2}) |N'|$.
    Consider a vertex $u \in L$ and a vertex $v \in \lset_0$.
    Since $N'$ is the neighbor set of $\lset_{2\tau}$, there is a vertex $v' \in \lset_{2\tau}$ that is a neighbor of $u$.
    In the end of round $\tau$, $\beta_{v'}$ is larger than $\beta_{v}$ by a factor of $(1 + \epsilon)^{2\tau}$.
    Thus, in the beginning of round $\tau$, $\beta_{v'}$ is larger than $\beta_v$ by at least a factor of $(1 + \epsilon)^{2\tau-2}$.
    Recall that $x_{u,v}$ and $x_{u,v'}$ are calculated proportional to the priority values of $v$ and $v'$ in the beginning of round $\tau$.
    Thus, $\vx_{u,v} \leq \vx_{u,v'} / (1+\epsilon)^{2\tau - 2}$.
    Since $\tau = \loglambda$, we have $\vx_{u,v} \leq \frac{\epsilon}{4\lambda}\cdot\vx_{u,v} \leq \frac{\epsilon}{4\lambda}$.
    That is, every edge $(u, v)$ with $u \in N'$ and $v \in \lset_0$ allocates at most an amount of $\frac{\epsilon}{4\lambda}$ to $\lset_0$.
    Consider the subgraph $G'$ of $G$ induced by $N' \cup \lset_0$.
    Since the arboricity of $G$ is bounded by $\lambda$ and $|N'| > |\lset_0|$, the number of edges in $G'$ is at most $\lambda \cdot (|N'| + |\lset_0|) < 2 \lambda |N'|$.
    Therefore, the total allocation from $N'$ to $\lset_0$, $\sum_{u \in N'} \sum_{v \in \lset_0} \vx_{u,v}$, is at most $2 \lambda |N'| \cdot \frac{\epsilon}{4\lambda} = \frac{\epsilon}{2} |N'|$.
    That is, the total allocation from $N'$ to vertices not in $\lset_0$ is at least $(1 - \frac{\epsilon}{2})|N'|$.
    This implies that $\sum_{j=1}^{2\tau} \sum_{v \in \lset_j} \alloc_v \geq (1 - \frac{\epsilon}{2}) |N'|$.

    By \cref{eqn:appendix-weight-bound-2}, $\mweight$ is at least
    
    \begin{equation*}
        \frac{1}{1+(k+2)\epsilon} \sum_{j=1}^{2\tau} \sum_{v \in \lset_j} \alloc_v
        \geq \frac{1}{1+(k+2)\epsilon} (1 - \frac{\epsilon}{2}) |N'|
        \geq \frac{1}{(1 + (k+2)\epsilon)(1+\epsilon)} |N'|.
    \end{equation*}
    Since $k\epsilon \leq 1$, the last term is at least $\frac{1}{1+(k+6)\epsilon}$. Hence, $|N'| \leq (1 + (k+6)\epsilon) \mweight$ also holds for this case.
\end{itemize}

Since we have $|N'| \leq (1 + (k+6)\epsilon) \mweight$ in both cases, the theorem holds.
\end{proof}

\cref{lem:equivalence,thm:appendix-apx-factor-1} imply the following.

\begin{theorem} Assume that we execute \cref{alg:matching} with $\epsilon \leq \frac{1}{4}$ and $\tau \geq \loglambda$. Then, with high probability, the algorithm finds a $(2 + 16 \epsilon)$-approximate fractional allocation.
\end{theorem}

\subsection{Performance guarantee on bipartite graphs} \label{sec:appendix-3}

In the following, we show that \cref{alg:det-thresh} has an approximation factor of $1 + (k + 14)\epsilon$ if $\frac{1}{k} \leq k_{v,r} \leq k$ holds for all $k_{v,r}$ and for some number $k$, $\epsilon \leq \frac{1}{k}$, and $\tau \geq \lognExact$.
By \cref{lem:equivalence}, this result implies that \cref{alg:matching} has an approximation factor of $1 + 18 \epsilon$ if $\tau \geq \lognExact$ and $\epsilon \leq \frac{1}{4}$.
The result is proven by modifying a theorem in \cite{AZM18}, which shows that \cref{alg:local} has an approximation factor of $(1 + O(\epsilon))$.

We first present lower bounds for $\mweight$. Let $g$ denote $\gap$.

\begin{lemma} \label{lem:appendix-weight-bound} For any two indices $1 \leq \ell$ and $\ell + g \leq \ell' \leq 2\tau$, we have

\begin{equation} \label{eqn:appendix-gap-bound} \mweight \geq \left( \frac{1}{1 + (k + 6)\epsilon} \right) \left(\sum_{j=0}^\ell  \mcap(\lset_j) + \left| \neigh \left( \bigcup_{j'=\ell'+1}^{2\tau} \lset_{j'} \right) \right| \right)
\end{equation}

\end{lemma}
\begin{proof} By \cref{lem:appendix-alloc-bound}, the vertices
$\lset_0, \lset_1, \dots, \lset_\ell$ almost fill up their capacities.
Therefore, they contribute at least $\frac{1}{1 + (k + 2)\epsilon}
\sum_{j=0}^{\ell} \mcap(\lset_j)$ which is larger than the first term of the lower bound.

The second term represents all neighbors of vertices in $\lset_{\ell'+1}, \dots, \lset_{2\tau}$. We first show that for each vertex $u \in L$ having a neighbor in $\bigcup_{j=\ell'+1}^{2\tau} \lset_j$, we have:

\begin{equation} \label{eqn:appendix-gap}
    \sum_{j=\ell+1}^{2\tau} \sum_{v \in \lset_j} \vx_{u, v} \geq (1 - \frac{\epsilon}{2}),
\end{equation}

That is, a large fraction of $u$ is assigned to vertices in $\lset_{\ell+1}, \dots, \lset_{2\tau}$.
Consider vertex $u \in L$ that is a neighbor of $v' \in \lset_{j'}$ for some $j' \geq \ell' + 1$.
Because $\ell'$ is at least $\ell + \gap$, we have $\beta_{v'} \geq \frac{2|R|}{\epsilon} \beta_v$ for any $v \in \lset_j$ with $j \leq l$.
Therefore $\vx_{u, v}$ is at most $\frac{\epsilon}{2|R|}$ times $\vx_{u, v'}$. Since there are at most $|R|$ vertices like $v$, we have

\begin{equation*}
    \sum_{j=0}^{\ell} \sum_{v \in \lset_j} \vx_{u, v} \leq \sum_{j=0}^{\ell} \sum_{v \in \lset_j} \frac{\epsilon}{2|R|} \cdot \vx_{u, v'} \leq \frac{\epsilon}{2} \vx_{u, v'} \leq \frac{\epsilon}{2}.
\end{equation*}
Hence, \cref{eqn:appendix-gap} holds, which indicates that each vertex in $\neigh(\bigcup_{j'=\ell'+1}^{2\tau} \lset_{j'})$ contributes at least $1 - \frac{\epsilon}{2}$ to $\sum_{j=\ell+1}^{2\tau} \alloc_v$. Using \cref{lem:appendix-alloc-bound}, at least $\frac{1}{1+(k+2)\epsilon}$ fraction of every such vertex will be counted toward $\mweight$. So in total, we get at least $\frac{1}{1+(k + 2)\epsilon} \cdot (1 - \frac{\epsilon}{2}) \geq \frac{1}{1+6\epsilon}$ fraction for each vertex in $\neigh(\bigcup_{j'=\ell'+1}^{2\tau} \lset_{j'})$. This concludes the proof.
\end{proof}

We prove the approximation factor of \cref{alg:det-thresh} by combining the lower and upper bounds.

\begin{lemma} \label{lem:appendix-apx-factor-2} \cref{alg:det-thresh} has an approximation factor of $1 + (k + 14) \epsilon$ if the input instance and parameters satisfy all of the following:
\begin{itemize}
    \item for all $v \in R$ and $0 \leq r \leq \tau$, $\frac{1}{k} \leq k_{a,r} \leq k$ for some number $k$,
    \item $\epsilon \leq \frac{1}{k}$, and
    \item $\tau \geq \gap$.
\end{itemize}
\end{lemma}
\begin{proof} There are two main gaps between the lower bound of \cref{eqn:appendix-gap-bound} and the upper bound of \cref{eqn:appendix-opt-bound}:
the $\frac{1}{1+(k+6)\epsilon}$ factor and the sum $\sum_{j=\ell+1}^{\ell'} \sum_{v \in \lset_j} \mcap_v$.
We show that the latter gap is small for some value of $\ell$ and
$\ell' = \ell + g$.
Recall that $g = \gap$.
Summing the gap over different values of $\ell$ yields
\begin{equation} \notag \textstyle
    \sum_{\ell=0}^{2\tau - g - 1} \sum_{j=\ell+1}^{\ell + g} \mcap(\lset_j)
    \leq g \sum_{j = 1}^{2\tau - 1} \mcap(\lset_j).
\end{equation}
Therefore, there exists an $\ell$ with $0 \leq \ell \leq 2\tau - g - 1$ such that its associated gap $\sum_{j=\ell+1}^{\ell+g} \mcap(\lset_j)$ is at most $\frac{g}{2\tau - g} \sum_{j=1}^{2\tau-1} \mcap(\lset_j) \leq \epsilon \sum_{j=1}^{2\tau-1} \mcap(\lset_j)$, where the last inequality holds when $\tau$ is at least $g / \epsilon = \lognExact$ and $\epsilon \leq 1$.

By \cref{eqn:appendix-weight-bound-1}, $\mweight$ is at least $\frac{1}{1+(k+2)\epsilon} \sum_{j=1}^{2\tau} \mcap(\lset_j)$. This means the gap associated for some $\ell$ is at most $\epsilon \cdot (1 + (k + 2)\epsilon) \mweight$. Using \cref{lem:appendix-weight-bound,eqn:appendix-opt-bound}, we have
\begin{equation}
\begin{split}
    \mweight & \geq \frac{1}{1+(k+6)\epsilon}(OPT - \epsilon (1 + (k + 2)\epsilon) \mweight) \\
    & = \frac{1}{1+(k+6)\epsilon} OPT - \epsilon \cdot \frac{1 + (k + 2)\epsilon}{1 + (k + 6)\epsilon} \cdot \mweight \\
    & \geq \frac{1}{1+(k+6)\epsilon} OPT - \epsilon \mweight,
\end{split}
\end{equation}
and hence $\mweight \geq \frac{1}{(1+(k+6)\epsilon)(1 + \epsilon)} OPT$,  yielding a final approximation factor of at least $1 + (k + 14)\epsilon$.
\end{proof}

\cref{lem:equivalence,lem:appendix-apx-factor-2} imply the following.

\begin{theorem} Assume that we execute \cref{alg:matching} with $\epsilon \leq \frac{1}{4}$ and $\tau \geq \lognExact$. Then, with high probability, the algorithm finds a $(1 + 18 \epsilon)$-approximate fractional allocation.
\end{theorem}

\section{\texorpdfstring{$(1 + \epsilon)$}{1+eps} Approximation of Allocation} \label{sec:framework}
In this section, we show that the framework in \cite[Section 4]{ghaffari2022massively} can be applied to find a $(1 +\eps)$-approximate allocation on low-arboricity graphs using our MPC algorithm in \cref{sec:mpc} as a subroutine. The framework was designed for finding a $(1+\eps)$-approximate \emph{$b$-matching}.

\begin{definition}[The $b$-matching problem]
In the $b$-matching problem, the input is a graph $G$ in which each vertex $v$ is associated with an integer $b_v$. A subset of edges is called a $b$-matching if each vertex $v$ has at most $b_v$ edges in this subset. The objective is to find a $b$-matching of the maximum size.
\end{definition}

Note that the allocation problem is a special case of the $b$-matching problem, in which a bipartition $(L, R)$ is given and all vertices $v \in L$ have $b(v) = 1$.

\subsection{Review of Ghaffari et al.'s framework}
Ghaffari et al.'s framework \cite{ghaffari2022massively} assumes access to a constant approximation algorithm $A$ for the $b$-matching problem. To find a $(1+\eps)$-approximate $b$-matching, they first invoke $A$ to find a constant approximate $b$-matching $M$.
Then, the framework iteratively improve $M$ via short augmenting walks.
In a fixed iteration, they first construct a layered graph $\cL$ whose vertex set is partitioned into layers $L_0, L_1, \dots, L_{k+1}$ for some integer $k = O(\eps^{-1})$.
The graph is constructed as follows.

\begin{enumerate}
    \item[] \textbf{Step 1:} Construct a graph $W$ as follows. For each vertex $v \in G$, add $b(v)$ copies of $v$, denoted by $v^1, v^2, \dots, v^{b(v)}$, into $V(W)$. For each edge $\{u, v\} \in M$, select a copy $u^i$ and a copy $v^j$ such that no $u^i$ nor $v^j$ has any edge from $M$ incident to it. Add edge $\{u^i, v^j\}$ to $E(W)$ and mark it as being in $M$. Note that the edge set of $W$ is a matching.
    \item[] \textbf{Step 2:} Each free vertex $v^i$ of $W$ (a vertex that is not adjacent to any edge) is assigned uniformly at random to $L_0$ or $L_{k+1}$.
    \item[] \textbf{Step 3:} Each matched edge $\{v, v'\}$ is assigned to one of the layers $L_1, \dots, L_k$ uniformly at random as an arc $(v, v')$ or $(v', v)$; the arc direction is also chosen randomly.
    \item[] \textbf{Step 4:} Let $H_i$ be the heads and $T_i$ be the tails of the arcs in $L_i$. For each unmatched edge $e = \{u, v\} \in E(G)$, choose an integer $i_e$ uniformly at random from the interval $0, 1, \dots, k$ and an orientation $\vec{e} = (u, v)$ or $\vec{e} = (v, u)$; wlog, assume $\vec{e} = (u, v)$. Then, we assign $e$ to be between a copy of $u$ which is in $H_{i_e}$ (if any) and a copy of $v$ which is in $T_{i_e+1}$ (if any).
    \item[] \textbf{Step 5:} For each layer $L_i$, contract all copies of a vertex $u$ in $H_i$ to a single vertex $u'$. Set $b'(u')$ to be the number of copies contracted in this step. Contract copies of each vertex $u$ in $T_i$ in a similar way.
\end{enumerate}

The key property of this construction is that any augmenting walk is preserved in the layered graph with probability $1/\exp(O(2k))$ \cite{ghaffari2022massively}.
The framework finds an preserved augmenting walk in the layer graph by executing $O(k^k)$ iterations.
Each iteration invokes the constant approximation algorithm $A$ to find a $b'$-matching between two adjacent layers $H_i$ and $T_{i+1}$.
Note that the input to each invocation of $A$ is a subgraph of $G$.
Hence, if $G$ has arboricity $\lambda$, the input to each invocation has arboricity at most $\lambda$.
We refer the reader to \cite[Section 4]{ghaffari2022massively} for more details of the framework.

\subsection{Solving the allocation problem}
Assume that the input to the framework is an instance of the allocation problem; that is, a bipartition $(L, R)$ is given and the $b$-value of each vertex $v \in L$ is $1$.
Then, the framework can be modified as follows.
In Step 2, we assign each free vertex in $L$ to $L_0$ and each free vertex in $R$ to $L_{k+1}$.
In Steps 3 and 5, we orient each unmatched edge from $L$ to $R$ and each matched edge from $R$ to $L$.
It is not hard to verify that the above modifications do not affect the correctness of the framework, because each short augmenting walk is still preserved in $\cL$ with probability $1/ \exp(O(2k))$.
With this modification, for each layer $L_i$, all vertices in $H_i$ are from $L$ and all vertices in $T_{i+1}$ are from $R$.
Hence, to find a $b$-matching between $H_i$ and $T_{i+1}$, we can invoke a constant approximation algorithm for the allocation problem.


\end{document}